\documentclass[10pt]{amsart}
\usepackage{amssymb}
\usepackage{esint}
\usepackage{braket,mleftright}
\usepackage{enumitem}
\usepackage{mathrsfs}
\usepackage{mathabx}
\usepackage{graphicx,caption,subcaption}
\usepackage{cite}
\usepackage[final=true]{hyperref}
\hypersetup{colorlinks=true,allcolors=[rgb]{0,0,0.6}}

\newcommand{\R}{\mathbb{R}}
\newcommand{\C}{\mathbb{C}}
\newcommand{\N}{\mathbb{N}}
\newcommand{\Rp}{\mathbb{R}^{+}}

\newcommand{\mrm}[1]{\mathrm{#1}}
\newcommand{\ol}[1]{\overline{#1}}
\newcommand{\co}{\colon}
\newcommand{\vrt}{\,\vert\,}
\newcommand{\lto}{\rightarrow}
\newcommand{\lmap}{\mapsto}
\newcommand{\abs}[1]{\lvert#1\rvert}
\newcommand{\norm}[1]{\lVert#1\rVert}

\newcommand{\ie}{\textit{i.e.}\;}
\newcommand{\eg}{\textit{e.g.}\;}
\newcommand{\cf}{\textit{cf.}\;}

\newcommand{\B}{\mathcal{B}}
\renewcommand{\H}{\mathfrak{H}}
\newcommand{\h}{\mathfrak{h}}
\newcommand{\K}{\mathfrak{K}}
\renewcommand{\L}{\mathfrak{L}}
\newcommand{\Hex}{\H^{\mrm{ex}}}
\newcommand{\HLex}{\H_L^{\mrm{ex}}}
\newcommand{\indx}{\mrm{I}}
\newcommand{\F}{\mathscr{F}}
\newcommand{\bJ}{\mathsf{J}}
\newcommand{\ho}{\hat{o}}
\newcommand{\he}{\hat{e}}
\newcommand{\e}{\mathsf{1}}
\newcommand{\Sym}{\mathfrak{S}}
\newcommand{\su}{\mathfrak{su}_2}
\newcommand{\dl}{\shortmid}
\newcommand{\U}{\mathfrak{U}}
\newcommand{\cgc}[6]{\begin{bmatrix}
#1 & #2 & #3 \\ #4 & #5 & #6\end{bmatrix}}
\newcommand{\indxl}{\mrm{I}_{\lambda}}
\newcommand{\sixj}[6]{\begin{Bmatrix}
#1 & #2 & #3 \\ #4 & #5 & #6\end{Bmatrix}}
\newcommand{\E}{\mathcal{E}}

\newcommand{\dfn}{:=}
\newcommand{\setm}{\smallsetminus}
\newcommand{\mcup}{\cup}
\newcommand{\mcap}{\cap}

\newcommand{\ot}{\otimes}
\newcommand{\hot}{\hat{\otimes}}
\newcommand{\op}{\oplus}

\DeclareMathOperator{\dom}{dom}
\DeclareMathOperator{\End}{End}
\DeclareMathOperator{\supp}{supp}
\DeclareMathOperator{\pv}{\fint}

\theoremstyle{plain}
\newtheorem{thm}{Theorem}[section]
\newtheorem{cor}[thm]{Corollary}
\newtheorem{lem}[thm]{Lemma}
\newtheorem{prop}[thm]{Proposition}
\newtheorem{hypo}[thm]{Hypothesis}

\theoremstyle{definition}
\newtheorem{rem}[thm]{Remark}
\newtheorem{rems}[thm]{Remarks}

\allowdisplaybreaks
\numberwithin{equation}{section}
\begin{document}
\title[Direct integral description of angulon]{Direct integral description of angulon}
\author{Rytis Jur\v{s}\.{e}nas}
\address{Vilnius University, Institute of Theoretical Physics and Astronomy, 
Saul\.{e}tekio ave.~3, 10222 Vilnius, Lithuania}
\email{Rytis.Jursenas@tfai.vu.lt}
\date{\today}
\begin{abstract} 
We propose a representation of angulon in which the angulon operator is decomposable
relative to the field of Hilbert spaces over the probability measure space, and the
probability measure corresponds to the total-number operator of phonons.
In this representation we are able to find the system of $N+1$ equations whose solutions 
form the eigenspace of the angulon operator, where $1\leq N<\infty$ is the number of phonon
excitations. Using this result we estimate the infimum of the spectrum. In the special 
case $N=1$, the lowest energy approximates to the value which is already known in the
literature. Our findings indicate that two-phonon excitations ($N=2$) contribute notably to
the energy of a molecule in superfluid $^4$He.
\end{abstract}
\maketitle
\section{Introduction}
In this paper the object of interest is the angulon, first introduced in
\cite{Schm-Lem-a} and later developed in \cite{Schm-Lem-b,Midya16}. Angulon represents a 
rotational analogue of polaron, actively studied in the context of solid state and atomic 
settings \cite{Griesemer,Devreese15}. However, non-commutativity and discrete energy 
spectrum inherent to quantum rotations renders the angulon physics substantially different 
compared to any of the polaron models. The concept of angulon considerably simplifies the 
problems involving an impurity possessing orbital angular momentum immersed into a bath of 
bosons. Thereby it paves the way for understanding cold molecules rotating inside 
superfluid helium nanodroplets \cite{ToenniesAngChem04} and ultracold gases 
\cite{Midya16,LemSchmidtChapter}, Rydberg electrons in Bose--Einstein condensates 
\cite{BalewskiNature13,SchmidtDem2016}, electronic excitations coupled to phonons in 
solids \cite{TowsPRL15}, and several other systems.

The angulon operator describes a linear-rotor molecule dressed by the bosonic field.
The operator consists of the kinetic energies of a molecule and the bosonic
field, as well as the interaction potential. The kinetic energy of a molecule is given
by $c\bJ^{2}$, where $c>0$ is the rotational constant and where the vector-valued operator 
$\bJ$ is called the angular momentum operator of a linear-rotor impurity. The bosonic part 
of the Hamiltonian for a molecular impurity describes bosons with some dispersion relation 
(typically $k^2$, where $k\geq0$) and with contact interactions. Using the Bogoliubov 
transformation which maps bosonic particles to quasi-particles, called phonons, one 
transforms the bosonic part to the Hamiltonian that describes phonons with a dispersion 
relation, say $\omega(k)$. One 
further expresses the creation/annihilation operators of phonons in the angular momentum 
basis so that the kinetic energy of the bosonic field is the free Hamiltonian, 
$\int\omega(k)n(k)dk$, of mass $\omega$. Here $n(k)$ is the sum of occupation numbers 
corresponding to the $\iota$th phonon state; $\iota$ runs over the set $\indx$ of 
pairs $(\lambda,\rho)$, where $\lambda\in\N_{0}$ is the \textit{phonon angular momentum} 
and $\rho\in\indxl\dfn\{-\lambda,-\lambda+1,\ldots,\lambda\}$. The integration is
performed over $\Rp\dfn[0,\infty)$.

In its original form, the interaction potential is given by 
\begin{equation}
2\Re\int \sum_{\lambda,\rho} U_\lambda(k)Y_{\lambda,-\rho}(\ho)\hat{b}_{\lambda\rho}(k)dk.
\label{eq:WSL}
\end{equation}
The angular momentum dependent coupling $U_\lambda(k)$ signifies the strength of the
potential and it depends on the microscopic details of the two-body interaction between
the impurity and the bosons \cite{Midya16}; typically, $U_\lambda(k)\equiv0$ for 
$\lambda>1$. The spherical harmonic $Y_{\iota}(\ho)$ is 
parametrized by the spherical angles $\ho$ of a molecule. The dependence on the molecular 
orientation gives rise to a nonzero commutator $[\bJ,Y_{\iota}(\ho)]$, which is  
calculated in a usual way \cite{Rudzikas,Rudz-Kan}. The bosonic creation and annihilation 
operators, $\hat{b}_\iota(k)$ and $\hat{b}_\iota(k)^*$, respectively, are defined as ordinary
operator fields in the sense of quadratic forms; the reader may refer to \cite{Berezin}, 
\cite[Sec.~5.2]{Bratteli}, and \cite[Sec.~X.7]{Reed} for a classic exposition.

Within the framework of the variational approach applied to the ansatz 
\cite{Schm-Lem-a} for the many-body quantum state with single-phonon excitation,
the eigenvalue $E<0$ of the angulon operator satisfies a Dyson-like 
equation 
\begin{equation}
E=cL(L+1)-\Sigma_L(E)
\label{eq:E-SL}
\end{equation}
where the so-called self-energy is defined as
\begin{equation}
\Sigma_L(E)\dfn\int\sum_{J,\lambda}\frac{2\lambda+1}{4\pi}
\cgc{L}{\lambda}{J}{0}{0}{0}^2
\frac{U_\lambda(k)^2}{cJ(J+1)+\omega(k)-E}dk.
\label{eq:Self-SL}
\end{equation}
and $\bigl[\begin{smallmatrix}L & \lambda & J \\ 0 & 0 & 0\end{smallmatrix}\bigr]$
is the Clebsch--Gordan coefficient for the tensor product $[L]\ot [\lambda]$ of
$SO_3$-irreducible representations \cite{Klimyk,Jucys-a,Jucys-b}.
The eigenvalue $E=E_L$ is labeled by the \textit{total angular momentum} $L$ obtained 
by reducing the tensor product $[J]\ot[\lambda]$; the highest weight $J\in\N_0$ of
the third component of the operator $\bJ=(J_1,J_2,J_3)$ is referred to as the
\textit{angular momentum of a linear-rotor impurity}. We have that $E$ is of multiplicity
$2L+1$.
\subsection{Problem setting}
In order to overcome the problem of adding a large number of angular momenta necessary
for the analysis of higher-order phonon excitations, the authors in \cite{Schm-Lem-b} 
make a one step further by introducing a rotation operator which is generated by the
collective angular momentum operator of the many-body bath. The angular momentum operator
has the highest weight $\Lambda\in\N_0$ and the total angular momentum $L$ is obtained
by reducing $[J]\ot[\Lambda]$ in a standard way. The transformation is useful when
the rotational constant $c\to0$, since in this regime the transformed angulon operator
can be diagonalized. Still, the eigenvalue is calculated by using the variational
ansatz based on \textit{single-phonon} excitations, though on top of the transformed 
operator.

In this paper we propose a scheme for treating higher-order phonon excitations
self-consistently. We do not rely on the limit of a slowly rotating impurity ($c\to0$),
nor we need an auxiliary variational ansatz. In fact, we construct a reference 
Hilbert space so that, to some extent, the ansatz is represented by an element of that 
space (\cf \eqref{eq:psiL} and \cite[Eq.~(3)]{Schm-Lem-a}). On the other hand, in our approach 
the number $N\in\N$ of phonon excitations is arbitrarily large but finite. The angular momentum 
$\Lambda$ is then obtained by reducing 
the tensor product of $n\in\{0,1,\ldots,N\}$ copies of $SO_3$-irreducible representations 
$[\lambda]$. The latter requires some angular momentum algebra, but the exposition
becomes elegant once we introduce certain symmetrization coefficients.
As a result, we come by the system of $N+1$ equations whose solutions form the eigenspace
of the angulon operator. 

To achieve our goals, we reformulate the definition of the angulon operator, which we
now show is equivalent to the original one.
\subsection{Description of the model}
We construct the angulon operator as a decomposable operator, $A=\int^{\op}A(k)\mu(dk)$, 
relative to the field $k\lmap\H(k)$ of Hilbert spaces over the probability measure space
$(\Rp,\F,\mu)$; the reader may refer to \cite[Sec.~12]{Schmudgen-2}, 
\cite[Chap.~II.2]{Dixmier} for basic definitions. Here and elsewhere, the direct 
integral is assumed over $\Rp$. The crucial point is that now we are able to put the sum
$n(k)$ of occupation numbers aside in a certain sense and yet to considerably simplify the 
spectral analysis of the angulon operator independently of $n(k)$. Assuming further that
$\mu\ll dk$, with the Radon--Nikodym derivative $\phi$ supported on the whole
$\ol{\Rp}=\Rp\mcup\{\infty\}$, we show that the greatest lower bound of $A$ does not 
depend on $\phi$, and hence on $k\lmap n(k)$.

$\bullet$ \textsl{Measure v. number operator.}
To see the connection between $n(k)$ and $\phi(k)$, let us consider a measurable field 
$k\lmap V(k)$ of Hilbert spaces over the probability measure space $(\Rp,\F,\mu)$. Let 
$\K=\int^{\op}\K(k)\mu(dk)$ be the direct integral of Hilbert spaces and let $\K(k)$ be 
the symmetric tensor algebra $S(V(k))$ equipped with an appropriate scalar product. Assume 
we are given two operators, $\omega$ and $n$, defined by measurable fields 
$k\lmap\omega(k)I(k)$ and $k\lmap n(k)$, respectively. Here $\omega(k)$ is the dispersion 
relation, as discussed above, $I(k)$ is the identity map in $\K(k)$, and $n(k)$ is the 
multiplication operator by $n\in\N_{0}$, provided that $n(k)$ acts on vectors from the 
$n$th symmetric tensor power $S^n(V(k))$. Given a unit vector $v(k)\in S^n(V(k))$, the 
scalar product $\braket{v,n\omega v}_\K$ in $\K$ reads $n\int\omega(k)\mu(dk)$. Now take 
$\mu\ll dk$ with the Radon--Nikodym derivative $\phi\in L^1(\Rp)$. Then, the scalar 
product reads $\int\omega(k)n_\phi(k)dk$ with $n_\phi(k)\dfn\phi(k)n$. We thus have 
obtained a formal analogue of the free bosonic Hamiltonian of mass $\omega$, but now the 
sum of phonon modes is given by $n_\phi(k)$ for a.e. $k$. The example suggests that the 
absolutely continuous (a.c.) probability measure corresponds to the total-number operator of 
phonons.

$\bullet$ \textsl{Coherent v. incoherent phonons.}
In our approach the creation and annihilation operators\footnote{The symbol 
$b_\iota(k)^*$, rather than $b_\iota(k)$, for denoting the annihilation operator seems to 
us more natural
because, as we shall see, the creation operator $b_\iota(k)$ defines the irreducible
tensor operator in the sense of Fano--Racah, while its adjoint $b_\iota(k)^*$ does not. 
Here we follow the notation of \cite{Rudzikas,Rudz-Kan}.}, $b_\iota$ and $b_\iota^*$, are 
decomposable operators relative to the field $k\lmap\K(k)$, and their commutator is given 
by
\begin{equation}
[b_\iota(k)^*,b_{\iota^\prime}(k)]=\delta_{\iota\iota^\prime}I(k)
\label{eq:0}
\end{equation}
where $\delta_{\iota\iota^\prime}\equiv\delta_{\lambda\lambda^\prime}
\delta_{\rho\rho^\prime}$
is the product of Kronecker symbols. At first glance one could expect from \eqref{eq:0}
that such a formulation of the bosonic field is suitable for the study of coherent phonons
only. However, the following argument shows that this is not the case.

Recall that $\Rp$ is the union of Borel sets $\sigma_d\in\F$ defined by $\{k\vrt F(k)=d\}$ 
for $d\in\N_0$ and some $\mu$-simple $F\co\Rp\lto\C$. That is, for every $k\in \sigma_d$, 
the Hilbert space $\K(k)$ is isomorphic to a separable Hilbert space, say $\h_d$, of dimension 
$d$. Let $j(k)$ be the isomorphism of $\K(k)$ onto $\h_d$ for $k\in\sigma_d$. To the operator 
$b_\iota(k)$ in $\K(k)$ corresponds the operator $b_{\iota,d}\dfn j(k)b_\iota(k)j(k)^*$ in $\h_d$ 
for every $k\in\sigma_d$. By \eqref{eq:0}, the operators $b_{\iota,d}^*$ and 
$b_{\iota^\prime,d^\prime}$ satisfy the commutation relation
\begin{equation}
[b_{\iota,d}^*,b_{\iota^\prime,d^\prime}]=\delta_{\iota\iota^\prime}\delta_{dd^\prime}I_d
\label{eq:0-1}
\end{equation}
where $I_d\dfn j(k)I(k)j(k)^*$, with $k\in\sigma_d$, is the identity map in $\h_d$. We see that
\eqref{eq:0-1} and the vanishing commutators 
\[[b_{\iota,d},b_{\iota^\prime,d^\prime}]=0, \quad 
[b_{\iota,d}^*,b_{\iota^\prime,d^\prime}^*]=0\]
together define an infinite-dimensional Heisenberg algebra. Thus we have that both the 
creation/annihilation operators defined via the operator-valued distributions (as in 
\eqref{eq:WSL}; see also \cite[Example~5.2.1]{Bratteli}) and the creation/annihilation operators 
defined as decomposable operators
relative to $k\lmap\K(k)$ are the representations of a centrally extended Lie algebra 
$\hat{\mathfrak{g}}=\mathfrak{g}\op\C1$, with $\mathfrak{g}$ a commutative Lie algebra. The 
commutation relations in $\hat{\mathfrak{g}}$ are defined by $[1,\hat{\mathfrak{g}}]=
[\hat{\mathfrak{g}},1]=0$ and $[x,y]=\braket{x,y}1$ for $x,y\in\mathfrak{g}$. 
$\hat{\mathfrak{g}}$ is a Lie algebra provided that
$[\cdot,\cdot]\co\hat{\mathfrak{g}}\times\hat{\mathfrak{g}}\lto
\hat{\mathfrak{g}}$ is an alternating bilinear map, and a
bilinear form $\braket{\cdot,\cdot}\co\mathfrak{g}\times\mathfrak{g}\lto\C$ is 
$\mathfrak{g}$-invariant. More on this topic can be found in \cite{Frenkel81}.

Now let $T$ be a representation of $\hat{\mathfrak{g}}$ in $\K$. Since $\K$ and the Fock space
$\mathfrak{F}(L^2(\Rp))$ are infinite-dimensional separable Hilbert spaces, there 
exists an isomorphism $\chi$ mapping $\K$ onto $\mathfrak{F}(L^2(\Rp))$, and hence 
$T_\chi(\hat{\mathfrak{g}})\dfn\chi T(\hat{\mathfrak{g}})\chi^{-1}$ is a 
representation of $\hat{\mathfrak{g}}$ in $\mathfrak{F}(L^2(\Rp))$. It is a 
classic result that, for a suitable choice of the bases, the two equivalent 
representations $T(\hat{\mathfrak{g}})$ and $T_\chi(\hat{\mathfrak{g}})$ are given by 
identical matrices; \ie $T(\hat{\mathfrak{g}})$ in the basis $B$ of $\K$ has the same 
matrix as $T_\chi(\hat{\mathfrak{g}})$ in the basis $\chi B$ of
$\mathfrak{F}(L^2(\Rp))$. In this respect, the direct integral description of angulon 
is equivalent to the original one.
\subsection{Main results}
We now briefly describe our main results, Theorems~\ref{thm:eigen}, \ref{thm:infnumL}.

$\bullet$ \textsl{Eigenspace.}
Since the angulon operator $A$ is defined by a measurable field $k\lmap A(k)$ relative to
$k\lmap\H(k)$, $E$ is an eigenvalue of $A$ iff $E$ is an eigenvalue of $A(k)$ for
$\mu$-a.e. $k$ (recall \eg \cite[Theorem~XIII.85(e)]{Reed-2}). In Sec.~\ref{sec:reduce} we show
that $A$ is the orthogonal sum of its parts $A_L$ acting in reducing subspaces $\H_L\subset
\H$. Thus, $E$ is an eigenvalue of $A_L(k)$ for $\mu$-a.e. $k$ and some $L$. 

An element $\psi$ of $\H$ is a square-integrable vector field $k\lmap\psi(k)$, and $\psi$ is in 
one-to-one correspondence with its coordinates calculated with respect to the field of 
orthonormal bases of $\H(k)$ (Theorem~\ref{thm:corth2} and \eqref{eq:hbasis}). When 
$\psi_{LM_L}(k)$, with $M_L\in\{-L,-L+1,\ldots,L\}$, is an eigenvector of $A_L(k)$ 
belonging to $E=E_{LM_L}$, the coordinates satisfy the relations as given in 
Theorem~\ref{thm:eigen}.

$\bullet$ \textsl{Greatest lower bound.}
Assuming that $\mu(dk)=\phi(k)dk$ with $\supp\phi=\ol{\Rp}$, we examine the infimum 
$\E=\E_{LM_L}$ of the spectrum of $A_L$. We show in Theorem~\ref{thm:infnumL} that 
$\E$ solves \eqref{eq:E-SL} for $N$ phonon excitations, but with different self-energy 
$\Sigma_L(\E)\geq0$; if however $\Sigma_L(\E)\leq0$ then $\E=cL(L+1)$. In particular, when 
$N=1$ (Corollary~\ref{cor:infnumL}), $\E$ approximates to \eqref{eq:E-SL}, 
\eqref{eq:Self-SL}. When $N=2$ and $L=0$ (Corollary~\ref{cor:infnumL2}), $\E\leq0$ solves 
$\E=-\Sigma_0(\E)$, where the self-energy is defined as
\begin{equation}
\Sigma_0(\E)\dfn\pv\sum_{\lambda}\frac{2\lambda+1}{4\pi}
\frac{U_\lambda(k)^2}{c\lambda(\lambda+1)+\omega(k)-\E-\epsilon_\lambda(\E,k)}dk
\label{eq:Self-N2L0}
\end{equation}
and $\pv$ is the Cauchy principal value of the integral over $\Rp$. Comparing  
\eqref{eq:Self-N2L0} with \eqref{eq:Self-SL} for $L=0$, we see that now the denominator 
contains an additional $\epsilon_\lambda(\E,k)$, which is defined as
\begin{equation}
\epsilon_\lambda(\E,k)\dfn\frac{1}{2\pi}(-1)^\lambda(2\lambda+1)U_\lambda(k)^2
\sum_\Lambda\frac{\cgc{\lambda}{\lambda}{\Lambda}{0}{0}{0}^2}{c\Lambda(\Lambda+1)+
2\omega(k)-\E}
\label{eq:Self-N2L0b}
\end{equation}
for a.e. $k$, and the sum runs over even numbers $\Lambda\in\{0,2,\ldots,2\lambda\}$. 
The lowest energy $\E\leq0$ for a molecule in superfluid $^4$He is shown in 
Fig.~\ref{fig:E012}. If we compared the energy with the curve $0_{0,0}$ in 
\cite[Fig.~2(a)]{Schm-Lem-a} for $N=1$, we would see that adding the two-phonon
excitations reflects in a significant increase of the lowest energy.
\begin{figure}[h!]
\centering
\includegraphics[width=.6\textwidth]{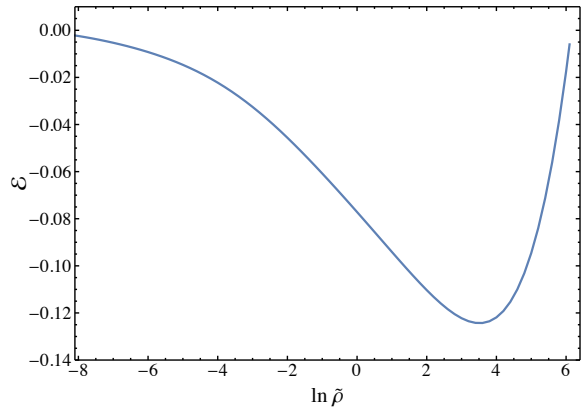}
\caption{The lowest energy $\E\leq0$ (in units of $c$) of the part $A_0$ of the angulon 
operator $A$ acting in reducing subspace labeled by $L=0$, when only two-phonon 
excitations are considered. The angulon operator describes a molecule in superfluid 
$^4$He. The energy is regarded as a function of the logarithmic superfluid density 
$\tilde{\rho}$. The functions $k\lmap\omega(k)$ and $k\lmap U_\lambda(k)$ and the 
parameters are adapted from \cite{Schm-Lem-a}.}
\label{fig:E012}
\end{figure}

We would like to point out that $\E$ might not belong to the numerical range $\Theta_L$ 
of $A_L$. Yet $\E$ lies at the bottom of the closure $\ol{\Theta_L}$.
\subsection{Outline of the paper}
Sec.~\ref{sec:nota} is of preliminary character. Here we sum up basic definitions that we 
use throughout the paper. In Sec.~\ref{sec:main} we give a precise definition of the 
angulon operator within the framework of the direct integral approach. Our definition 
relies on the hypothesis that the functions $k\lmap\omega(k)$ and 
$k\lmap\sum_\lambda(2\lambda+1)^{3/2}U_\lambda(k)$ are $\mu$-essentially bounded. The 
hypothesis considerably simplifies the presentation, since in this case the angulon 
operator is a self-adjoint decomposable operator defined on the domain $\dom\bJ^2\ot\K$, 
where $\dom\bJ^2$ is the maximal domain of $\bJ^2$. In Sec.~\ref{sec:ratation} we show 
that the angulon operator $A$ is unitarily equivalent to the $SO_3$-scalar tensor operator 
$A^\prime$, and we further identify $A$ with $A^\prime$. In Sec.~\ref{sec:scfps} we 
introduce the symmetric coefficients of fractional parentage (SCFPs) and describe their 
properties. With the help of the SCFPs we construct the field of orthonormal bases
that transform under rotations in $\K(k)$ as an irreducible tensor operator of rank
$\Lambda$ (Theorem~\ref{thm:corth2}). Then, reducing the tensor product
$[J]\ot[\Lambda]\lto[L]$, we find reducing subspaces $\H_L$ for the angulon operator
(Secs.~\ref{sec:orth-bases} and \ref{sec:reduce}). We study the eigenvalues in
Sec.~\ref{sec:eigen}. Using Theorem~\ref{thm:eigen} we examine the numerical range
in Sec.~\ref{sec:num}, and we summarize our findings in Theorem~\ref{thm:infnumL}
and the subsequent corollaries. 
\section{Preliminaries and notation}\label{sec:nota}
Here and elsewhere, $V$ is the direct integral $\int^{\op}V(k)\mu(k)$ of a measurable
field $k\lmap V(k)$ of Hilbert spaces over the probability measure space $(\Rp,\F,\mu)$.
Measurable fields are assumed to be $\mu$-measurable. The integral $\int^{\op}$ ($\int$) 
means the integral over $\Rp$ unless specified otherwise.
The scalar product is indexed by the Hilbert space in which it is defined; \eg 
$\braket{\cdot,\cdot}_V$ refers to the scalar product in $V$. 
Fix $(e_\iota(k))_{\iota\in\indx}$ as an orthonormal basis of 
$V(k)$; the sequence $(e_\iota)_{\iota\in\indx}$ of measurable vector fields
$k\lmap e_\iota(k)$ is a measurable field of orthonormal bases 
\cite[Proposition~II.1.4.1]{Dixmier}.

Let $\iota=(\iota_1,\ldots,\iota_n)\in\indx^n$. To the basis vector 
$e_\iota(k)\dfn \ot_{i=1}^n e_{\iota_i}(k)$ of the $n$th tensor power $T^n(V(k))$ 
corresponds the basis vector $\he_\iota(k)\dfn e_\iota(k)+I(V(k))$ of the $n$th symmetric 
tensor power $S^n(V(k))$. The ideal $I(V(k))$ of the tensor algebra $T(V(k))$ is 
generated by the elements of the form $w(k)\ot v(k)-v(k)\ot w(k)$ with 
$w(k),v(k)\in V(k)$. When $n=0$, $\iota=\iota_0$ is empty and 
$e_{\iota_0}(k)\dfn\e(k)\in T^0(V(k))$ is the unit element. For notational convenience, 
we usually write $\he_\iota(k)$ in the form $\hot_{i=1}^n e_{\iota_i}(k)$, where the
symbol $\hot$ alludes to the symmetrized tensor product. With this notation
$\he_{\iota\iota^\prime}(k)$, with $\iota^\prime\in\indx^{n^\prime}$, implies
$\he_\iota(k)\hot\he_{\iota^\prime}(k)$, and vice verse. According to
\cite[Proposition~II.1.8.10]{Dixmier}, the field $k\lmap e_\iota(k)$ is measurable,
hence so is the field $k\lmap \he_\iota(k)$.

The action of the symmetric group $\Sym_n$ on a sequence $\iota\in\indx^n$ of length $n$
is defined as $\pi\cdot\iota=(\iota_{\pi(1)},\ldots,\iota_{\pi(n)})$ with $\pi\in\Sym_n$
and $n\in\N$; when $n=0$, $\pi\cdot\iota_0=\iota_0$. We have an important, although 
obvious, relation $\he_{\pi\cdot\iota}(k)=\he_\iota(k)$.

Assume that $\hot_{i=1}^nv_i(k),\hot_{i=1}^nu_i(k)\in S^n(V(k))$. The symmetric tensor 
algebra $S(V(k))$ completed with the norm that is induced by the scalar product
\begin{equation}
\braket{\hot_{i=1}^nv_i(k),\hot_{i=1}^nu_i(k)}_{\K(k)}\dfn
\frac{1}{n!}\sum_{\pi\in \Sym_n}\prod_{i=1}^n\braket{v_i(k),u_{\pi(i)}(k)}_{V(k)}
\label{eq:scalarSV}
\end{equation}
is the Hilbert space, denoted by $\K(k)$ \cite[Secs.~IV.2.4, V.3]{BourbakiTVS};
$\K=\int^{\op}\K(k)\mu(dk)$ is the direct integral of a measurable field 
$k\lmap\K(k)$ of Hilbert spaces. 

An element $v$ of $\K$ is a square-integrable vector field $k\lmap v(k)$ with the value
\begin{equation}
v(k)=\sum_{n\in\N_0}\sum_{\iota\in\indx^n}c_\iota(v(k))\he_\iota(k), \quad
c_\iota(v(k))\dfn\braket{\he_\iota(k),v(k)}_{\K(k)}
\label{eq:v}
\end{equation}
defined for $\mu$-a.e. $k$. Indeed, \eqref{eq:v} implies 
\[\braket{\he_{\iota^\prime}(k),v(k)}_{\K(k)}=
\sum_{n\in\N_0}\sum_{\iota\in\indx^n}c_\iota(v(k))\braket{\he_{\iota^\prime}(k),
\he_\iota(k)}_{\K(k)}\]
with $\iota^\prime\in\indx^{n^\prime}$. But, using \eqref{eq:scalarSV}
\[\braket{\he_{\iota^\prime}(k),\he_{\iota}(k)}_{\K(k)}=
\frac{\delta_{n^\prime n}}{n^\prime!}\sum_{\pi\in\Sym_{n^\prime}}
\delta_{\iota^\prime,\pi\cdot\iota}\]
and $\delta_{\iota^\prime\iota}$ is the product of 
Kronecker symbols $\delta_{\lambda_1^\prime\lambda_1}\delta_{\rho_1^\prime
\rho_1}\cdots\delta_{\lambda_{n^\prime}^\prime\lambda_{n^\prime}}
\delta_{\rho_{n^\prime}^\prime\rho_{n^\prime}}$, abbreviated as
$\delta_{\iota_1^\prime\iota_1}\cdots\delta_{\iota_{n^\prime}^\prime\iota_{n^\prime}}$.
Now use the fact that the coordinate $c_\iota(v(k))$ of $v(k)\in\K(k)$ is invariant under 
the action of $\Sym_n$.

The function $k\lmap(c_\iota(v(k)))_{\iota\in\indx^{\N_0}}$ is of class 
$L^2(\Rp,\mu;\ell^2(\indx^{\N_0}))$. Thus, \eqref{eq:v} defines a one-to-one 
correspondence between the latter Hilbert space and $\K$.

The algebra $S(V(k))$ is canonically isomorphic to the commutative polynomial algebra 
which, by definition, consists of elements indexed by finite sequences 
\cite[Sec.~II.1, Theorem~1.1]{Chevalley}. Therefore, there exists a finite 
$N\in\N_0$ such that $c_\iota(v(k))$ vanishes identically for $n>N$:
\begin{equation}
(c_\iota(v(k))\vrt \iota\in\indx^n)_{n\in\N_0}\in c_{00}(\N_0).
\label{eq:c00}
\end{equation}
Let us recall that $c_{00}$ is the space of eventually vanishing sequences.
\section{Angulon operator. Main definition}\label{sec:main}
\subsection{Kinetic energy of phonons}\label{sec:creatanih}
For every $k$, the \textit{creation operator} in the Hilbert space $\K(k)$ is defined by 
\begin{equation}
b_{\iota_1}(k)v(k)=\sum_{n\in\N_0}\sqrt{n+1}e_{\iota_1}(k)\hot v_n(k) 
\label{eq:b} 
\end{equation}
where $\iota_1\in\indx$ and $v(k)=\sum_nv_n(k)\in\K(k)$; $v_n(k)$ is found from 
\eqref{eq:v}. Using definition \eqref{eq:scalarSV}, 
$\norm{e_{\iota_1}(k)\hot v_n(k)}_{\K(k)}\leq\norm{v_n(k)}_{\K(k)}$. 
Using \eqref{eq:c00}, 
$b_{\iota_1}(k)\in\B(\K(k))$ is bounded in $\K(k)$:
\begin{equation}
\norm{b_{\iota_1}(k)v(k)}_{\K(k)}\leq C_v\norm{v(k)}_{\K(k)}
\label{eq:bnorm}
\end{equation}
where the constant $C_v\geq1$ depends only on a field $v$.

The formal adjoint of $b_{\iota_1}(k)$ coincides with the adjoint 
$b_{\iota_1}(k)^*$, which is called the \textit{annihilation operator}. 
Using \eqref{eq:scalarSV} and \eqref{eq:b}
\begin{equation}
b_{\iota_1}(k)^*v(k)=\sum_{n\in\N}\sqrt{n}\sum_{\iota\in\indx^{n-1}}
c_{\iota_1\iota}(v(k))\he_\iota(k).
\label{eq:b+}
\end{equation}
Using \eqref{eq:scalarSV} and \eqref{eq:b+}, the definition of the coordinate
$c_{\iota_1\iota}(v(k))$ in \eqref{eq:v}, and then applying the Cauchy--Schwarz
inequality and relation \eqref{eq:c00}
\begin{equation}
\norm{b_{\iota_1}(k)^*v(k)}_{\K(k)}\leq C_v^\prime\norm{v(k)}_{\K(k)}
\label{eq:b+norm}
\end{equation}
where the constant $C_v^\prime\geq0$ depends only on $v$.

Using \eqref{eq:b} and \eqref{eq:b+}
\begin{equation}
\sum_{\iota_1\in\indx}b_{\iota_1}(k)b_{\iota_1}(k)^*=n(k) \quad \text{where} \quad 
n(k)v(k)=\sum_{n\in\N_0}nv_n(k)
\label{eq:n}
\end{equation}
and the sum over $\iota_1$ is understood as a strong limit of
partial sums; the creation/ahhinilation operators satisfy the commutation 
relations
\begin{equation}
[b_{\iota_1}(k)^\#,b_{\iota_2}(k)^\#]=0, \quad
[b_{\iota_1}(k)^*,b_{\iota_2}(k)]=\delta_{\iota_1\iota_2}I(k)
\label{eq:ccr}
\end{equation}
for $\iota_1,\iota_2\in\indx$. Here $b_{\iota_1}(k)^\#$ denotes either $b_{\iota_1}(k)$ 
or $b_{\iota_1}(k)^*$.

By definition, for every $v\in\K$, the field $k\lmap b_{\iota_1}(k)^\#v(k)$ is
measurable. Thus the field $k\lmap b_{\iota_1}(k)^\#$ of bounded operators is
measurable, and it defines a bounded decomposable operator 
$b_{\iota_1}^\#=\int^{\op}b_{\iota_1}(k)^\#\mu(dk)$ \cite[Definition~II.2.3.2]{Dixmier},
called the creation/annihilation operator in $\K$. Let $n(k)\in\B(\K(k))$ be as in
\eqref{eq:n}. A bounded self-adjoint operator $n=\int^{\op}n(k)\mu(dk)$ is called the 
\textit{number operator}. Interpreting $b_{\iota_1}^\#$ as the creation/annihilation 
operator of the $\iota_1$th phonon state, one concludes that the number operator $n$ 
is the sum of occupation numbers $b_{\iota_1}b_{\iota_1}^*$ of phonon states, \ie
$n$ is the total-number operator.

Let $k\lmap\omega(k)$ be a continuous nonnegative function such that the field
$k\lmap\omega(k)I(k)$ is measurable. A diagonalizable operator
\[\omega=\int^{\op}\omega(k)I(k)\mu(dk)\] 
defined on the domain 
\[\dom\omega=\{v\in\K\vrt \omega(k)v(k)\in\K(k)\;\text{$\mu$-a.e.,}
\int\norm{\omega(k)v(k)}_{\K(k)}^2\mu(dk)<\infty\}\]
is self-adjoint. We call the operator $n\omega$ on $\dom\omega$ the 
\textit{kinetic energy of phonons}.

For future reference, let us remark the following.
\begin{prop}\label{rem:1}
When $k\lmap\omega(k)$ is of class $L^\infty(\Rp,\mu)$, $\dom\omega=\K$.
\end{prop}
This is merely due to H\"{o}lder inequality \cite[Theorem~2.4 and Corollary~2.5]{Adams03}.
\subsection{Kinetic energy of a molecule}
Let $\L_0$ be a separable complex Hilbert space with an orthonormal basis $(w_{JM})$;
$M\in\indx_J\dfn \{-J,\ldots,J\}$, $J\in\N_0$. $\L_0$ admits a decomposition, 
$\op_J\L_{0,J}$, where 
each Hilbert space $\L_{0,J}$ has the basis $(w_{JM})_M$. The kinetic energy of a 
molecule is described as follows. Let $J_x$, with $x\in\{1,2,3\}$, be the 
$\su$-algebra representation on $\L_{0,J}$ corresponding to the $x$th generator of $\su$.
The action of $J_x$ on $\L_{0,J}$ is defined as in \cite[Eq.~(14.5)]{Rudzikas};
so $J_x$ is simple (or else irreducible). Considering $\L_{0,J}$ as a vector space,
one extends $J_x\in\End \L_{0,J}$ to $\L_0$ by linearity. The extensions form
the vector-valued operator $\bJ=(J_1,J_2,J_3)$ referred to as the angular momentum
operator of a linear-rotor impurity. The operator $\bJ^2\equiv\bJ\bJ$ defined on its 
maximal domain reads
\begin{equation}
\bJ^2=\sum_{J\in\N_{0}}\sum_{M\in\indx_J}J(J+1)w_{JM}^{\dl}\ot w_{JM}, \quad
\dom\bJ^2=\{w\in\L_0\vrt\bJ^2w\in\L_0\}
\label{eq:J2}
\end{equation}
where $(w_{JM}^\dl)$ is the adjoint basis. One identifies
$w_{JM}^\dl\ot w_{JM}\co\L_0\lto\L_0$ with a rank one projection
$\braket{w_{JM},\cdot}_{\L_0}w_{JM}$. The operator $\bJ^2$ is written in its spectral 
representation, hence it is self-adjoint. For $c>0$, $c\bJ^2$ is called the 
\textit{kinetic energy of a molecule}.
\subsection{Phonon excitations}
Let $\H\dfn\L_0\ot\K$ be the Hilbert tensor product. Consider the constant field 
$k\lmap \L(k)$ of complex Hilbert spaces corresponding to $\L_{0}$. That is, 
$\L=\int^{\op}\L(k)\mu(dk)$ coincides with $L^{2}(\Rp,\mu;\L_{0})$. Then 
$\int^{\op}\H(k)\mu(dk)$, with $\H(k)\dfn\L(k)\ot\K(k)$, is isomorphic to $\H$ 
\cite[Sec.~II.1.8]{Dixmier} and therefore will be identified with $\H$ in what follows.

Let $\lambda\in\N_0$ and let $U_\lambda\co k\lmap U_{\lambda}(k)$ be a nonnegative 
measurable function. Let us define
\begin{subequations}\label{eq:User}
\begin{align}
U&\co k\lmap U(k)\dfn\sum_{\lambda\in\N_0}(2\lambda+1)^{3/2}U_{\lambda}(k), \\
\U&\dfn\{v\in\K\vrt U(k)v(k)\in\K(k)\;\text{$\mu$-a.e.,}
\int\norm{U(k)v(k)}_{\K(k)}^2\mu(dk)<\infty\}.
\end{align}
\end{subequations}
\begin{rem}
Similar to Proposition~\ref{rem:1}, if $U\in L^\infty(\Rp,\mu)$ then $\U=\K$.
\end{rem}
Let $I_0$ be the identity map in $\L_0$ and consider the operator in $\H$ defined by 
\begin{subequations}\label{eq:Qoper}
\begin{align}
Q\dfn&\sum_{\lambda\in\N_0}\sum_{\rho\in\indxl}
Y_{\lambda,-\rho}(\ho)I_0\ot U_{\lambda\rho}, \quad \dom Q=\{\psi\in\H\vrt Q\psi\in\H\}
\\
U_{\lambda\rho}\dfn&\int^{\op} U_{\lambda\rho}(k)\mu(dk), \quad
U_{\lambda\rho}(k)\dfn U_{\lambda}(k)b_{\lambda\rho}(k).
\end{align}
\end{subequations}
The infinite sum over $\lambda$ is understood as a strong limit of partial sums.
The spherical harmonic $Y_{\lambda\rho}(\ho)$ depends on the molecular orientation given 
in terms of the spherical angles $\ho$; \ie the commutator $[J_x,Y_{\lambda\rho}(\ho)]$ 
is calculated using \cite[Eq.~(14.13)]{Rudzikas}, \cite[Eq.~(3.53)]{Rudz-Kan}.
\begin{lem}\label{lem:Q}
$\L_0\ot \U\subseteq\dom Q$.
\end{lem}
\begin{proof}
Let $w\ot v\in\H$. By \eqref{eq:Qoper}, and then using the Cauchy--Schwarz inequality
\[\norm{Q(w\ot v)}_{\H}^2\leq \int(\sum_{\lambda,\rho}
\norm{Y_{\lambda,-\rho}w}_{\L_0}\norm{U_{\lambda\rho}(k)v(k)}_{\K(k)})^2\mu(dk). \]
Using $\norm{Y_{\lambda\rho}w}_{\L_0}\leq\sqrt{(2\lambda+1)/(4\pi)}\norm{w}_{\L_0}$, 
it follows from \eqref{eq:bnorm} and \eqref{eq:User} that
\[\norm{Q(w\ot v)}_{\H}\leq(4\pi)^{-1/2}C_v\norm{w}_{\L_0}
(\int\norm{U(k)v(k)}_{\K(k)}^2\mu(dk))^{1/2}.\]
Therefore $w\ot v\in\L_0\ot\U\Longrightarrow w\ot v\in\dom Q$.
\end{proof}
Lemma~\ref{lem:Q} suggests that our discussion considerably simplifies provided that
\begin{hypo}\label{hypo:1}
$k\lmap\omega(k)$ and $k\lmap U(k)$ are of class $L^\infty(\Rp,\mu)$.
\end{hypo}
Yet Hypothesis~\ref{hypo:1} encloses the physically interesting cases.
We thus have that $Q\in\B(\H)$ is a bounded operator in $\H$. Then, the operator
\begin{equation}
W\dfn Q+Q^*
\label{eq:W}
\end{equation}
is bounded, self-adjoint, and decomposable relative to the field $k\lmap\H(k)$. The 
operator $W$ describes phonon excitations and is referred to as the 
\textit{impurity-boson interaction potential} (\cf \eqref{eq:WSL}).
\subsection{Angulon operator}
Using Hypothesis~\ref{hypo:1}, the operator
\begin{equation}
A^0\dfn c\bJ^2\ot I+I_0\ot n\omega
\label{eq:A0}
\end{equation}
is a self-adjoint operator defined on $\dom A^0=\dom\bJ^2\ot\K$. The operator
\begin{equation}
A\dfn A^0+W, \quad \dom A=\dom A^0
\label{eq:A}
\end{equation}
is called the \textit{angulon operator}. Since
$W=\int^{\op}W(k)\mu(dk)$ is bounded decomposable and $A^0=\int^{\op}A^0(k)\mu(dk)$ is 
self-adjoint decomposable, the angulon operator 
$A=\int^{\op}A(k)\mu(dk)$ is decomposable \cite{Lance} with $A(k)\dfn A^0(k)+W(k)$. Since 
$\dom(A^0+W)=\dom A^0$, the operator sum $A$ of two self-adjoint operators $A^0$ and $W$ 
is self-adjoint.
\section{Rotated angulon operator}\label{sec:ratation}
By definition~\eqref{eq:A}, the angulon operator $A=A(\ho)$ depends on the 
molecular orientation, as defined by the spherical angles $\ho=(\vartheta,\varphi)$.
Here we show that, with a suitable choice of $\ho^\prime$, the operator
$A^\prime=A(\ho^\prime)$ is unitarily equivalent to $A$ and it commutes with rotations in 
$\H$. It is this result that eventually implies that the angulon is an eigenstate of
the total angular momentum $L$ of the system.

By construction, the vector space $V(k)$ admits a decomposition 
$\op_\lambda V_\lambda(k)$. Similar to $J_x$ on $\L_{0,J}$, we let $\Lambda_x$,
with $x\in\{1,2,3\}$, be a simple $\su$-algebra representation on $V_\lambda(k)$.
We denote by $M(R)$ be the rotation operator generated by $\{\Lambda_1,\Lambda_3\}$
\cite[Eq.~(3.8)]{Rudz-Kan}.
Here $R=R(\Phi,\Theta,\Psi)$ is a rotation matrix of $SO_3$ parametrized by Euler angles.

Using \eqref{eq:b}, $e_\iota(k)=b_\iota(k)\e(k)$ for $\iota\in \indx$. Therefore 
$b_{\lambda\rho}(k)$ is the $\rho$th component of a rank-$\lambda$ irreducible tensor 
operator $b_\lambda(k)$ in the sense of Fano--Racah \cite{Rudz-Kan,Rudzikas}. Unlike
$b_\lambda(k)$, the adjoint operator $b_{\lambda\rho}(k)^*$ does not transform as the 
irreducible tensor operator, since there arises an additional phase factor in
the complex conjugate of the Wigner $D$-function $D_{\rho\rho^\prime}^\lambda(R)$, \ie
$\ol{D_{\rho\rho^\prime}^\lambda(R)}=
(-1)^{\rho-\rho^\prime}D_{-\rho,-\rho^\prime}^\lambda(R)$ for $\rho,\rho^\prime\in\indxl$
(see \eg \cite[Eqs.~(2.19) and (2.21)]{Rudzikas}). However, the operator 
\begin{equation}
b_{\lambda\rho}(k)^{\sim}\dfn(-1)^{\lambda-\rho}b_{\lambda,-\rho}(k)^*
\label{eq:btilde}
\end{equation}
does define the $\rho$th component of a rank-$\lambda$ irreducible tensor operator 
$b_{\lambda}(k)^{\sim}$. Here the extra phase $(-1)^\lambda$ is introduced for 
convenience.

The tensor product module $\L_{0,J}\ot V_\lambda(k)$ is the tensor product vector space
with the action of $\su$ determined by the relation
\begin{equation}
L_x(w\ot v(k))=J_x w\ot v(k)+w\ot \Lambda_xv(k)
\label{eq:Lh}
\end{equation}
for $w\ot v(k)\in \L_{0,J}\ot V_\lambda(k)$. Similar to $M(R)$, let $L(R)$ and $K(R)$ 
be the rotation operators generated by $\{J_1,J_3\}$ and $\{L_1,L_3\}$, respectively.
Using \eqref{eq:Lh}
\begin{equation}
K(R)=L(R)\ot M(R).
\label{eq:KLM}
\end{equation}
By linearity, a unitary $K(R)$ extends to the operator in $\L_0\ot V(k)$.

Given two simple modules, $V_\lambda(k)$ and $V_{\lambda^\prime}(k)$, 
the tensor product module $V_\lambda(k)\ot V_{\lambda^\prime}(k)$ is the tensor product
vector space with the action of $\su$ determined by the condition
\[\Lambda_x(v(k)\ot u(k))=\Lambda_xv(k)\ot u(k)+v(k)\ot\Lambda_xu(k)\]
for $v(k)\ot u(k)\in V_\lambda(k)\ot V_{\lambda^\prime}(k)$. The tensor product of
finitely many $\su$-modules is defined analogously. Using the latter, $K(R)$ extends to
the unitary operator in $\L_0\ot \K(k)$. There exists a unique isomorphism mapping 
$w\ot v\in\L_0\ot\K$ into the field $k\lmap w \ot v(k)$; so the operator 
\[\mathcal{K}(R)\dfn\int^{\op}K(R)\mu(dk)\]
is a unitary decomposable operator in $\H$.

Let $R_{\ho}\dfn R(\varphi+\pi/2,\vartheta,0)$. According to \cite{Rudzikas},
this particular $R_{\ho}$ is useful in that it gives
\begin{equation}
Y_{\lambda\rho}(\ho)=\sqrt{(2\lambda+1)/(4\pi)}D_{\rho0}^\lambda(R_{\ho}) \quad
\text{and} \quad
\ol{Y_{\lambda\rho}(\ho)}=Y_{\lambda,-\rho}(\ho^\prime)
\label{eq:YD}
\end{equation}
for $\ho^\prime\equiv \ho\pm(0,\pi)\!\!\mod (0,2\pi)$. Let us define
\begin{equation}
A^\prime=A^\prime(\ho)\dfn\mathcal{K}(R_{\ho})A(\ho)\mathcal{K}(R_{\ho})^*
\label{eq:arotate}
\end{equation}
and call $A^\prime$ the \textit{rotated angulon operator}. We have that
\begin{thm}\label{thm:KR}
$A^\prime(\ho)=A(\ho^\prime)$.
\end{thm}
An immediate consequence of Theorem~\ref{thm:KR} is that $A^\prime$ is the 
$SO_3$-scalar tensor operator. To see this, put $Y_{\lambda\rho}(\ho^\prime)=(-1)^\rho
Y_{\lambda\rho}(\ho)$ in $A(\ho^\prime)$ and sum over $\rho\in\indxl$. Applying the rules 
for reducing the tensor product $[\lambda]\ot[\lambda]$ of $SO_3$-irreducible 
representations
\[\sum_{\rho\in\indxl}Y_{\lambda,-\rho}(\ho^\prime)I_0\ot b_{\lambda\rho}(k)=
(-1)^\lambda\sqrt{2\lambda+1}\{Y_\lambda(\ho)I_0\ot b_\lambda(k)\}_0\]
where $\{Y_\lambda(\ho)I_0\ot b_\lambda(k)\}_0$ is a rank-0 $SO_3$-irreducible tensor 
operator. With the help of \eqref{eq:btilde}, the same is done for the adjoint operator. 
By \eqref{eq:Qoper}--\eqref{eq:A}
\begin{subequations}\label{eq:a1prime}
\begin{align}
A^\prime=&A^0+W^\prime, \quad W^\prime=\int^{\op}W^\prime(k)\mu(dk) \quad
\text{with} \quad \\
W^{\prime}(k)\dfn&\sum_{\lambda\in\N_0}\sqrt{2\lambda+1}U_{\lambda}(k)
((-1)^\lambda\{Y_\lambda \ot b_\lambda(k)\}_{0}
+\{Y_\lambda \ot b_\lambda(k)^\sim\}_{0})
\end{align}
\end{subequations}
where $Y_\lambda=Y_\lambda(\ho)$ and, for simplicity, we omit $I_0$.
\begin{proof}[Proof of Theorem~\ref{thm:KR}]
Since $A^0$ is invariant under rotations in $\H$, it suffices to examine 
$W$. Using the tensor structure of measurable fields $k\lmap b_{\lambda\rho}(k)$ and 
$k\lmap b_{\lambda\rho}(k)^{\sim}$, one writes $W(k)$ as a sum of rank-$L$ tensor 
operators
\begin{equation}
W(k)=\sum_{\lambda,L}U_\lambda(k)x_{\lambda L}
(\{Y_\lambda(\ho)\ot b_\lambda(k)\}_{L0}+(-1)^\lambda
\{Y_\lambda(\ho)\ot b_\lambda(k)^{\sim}\}_{L0})
\label{eq:a1L}
\end{equation}
where $L\in\{0,2,\ldots,2\lambda\}$ and $x_{\lambda L}$ is the sum of
Clebsch--Gordan coefficients:
\[x_{\lambda L}\dfn\sum_\rho\cgc{\lambda}{\lambda}{L}{-\rho}{\rho}{0}.\]
Let $b_{\lambda\rho}(k)^\#$ denote either $b_{\lambda\rho}(k)$ or
$b_{\lambda\rho}(k)^{\sim}$. Using \eqref{eq:KLM}
\begin{align*}
K(R)^*(Y_{\lambda,-\rho}(\ho)&I_0\ot b_{\lambda\rho}(k)^\#)K(R) \\
&=\sum_{\rho^\prime}D_{-\rho,\rho^\prime}^\lambda(R)
Y_{\lambda\rho^\prime}(\ho)I_0\ot 
\sum_{\rho^{\prime\prime}}D_{\rho\rho^{\prime\prime}}^\lambda(R)
b_{\lambda\rho^{\prime\prime}}(k)^\#.
\end{align*}
Now choose $R=R_{\ho}^{-1}$. Using \eqref{eq:YD}, the sum over 
$\rho^\prime\in\indxl$ gives a factor $\delta_{\rho0}$ and 
\begin{align}
K(R_{\ho})\{Y_\lambda(\ho)I_0&\ot b_\lambda(k)^\#\}_{L0}K(R_{\ho})^* \nonumber \\
=&\cgc{\lambda}{\lambda}{L}{0}{0}{0}
\sum_{L^\prime}x_{\lambda L^\prime}
\{Y_\lambda(\ho^\prime)I_0\ot b_\lambda(k)^\#\}_{L^\prime0}.
\label{eq:krbb}
\end{align}
Apply \eqref{eq:krbb} to \eqref{eq:a1L} and calculate the sum over $L$, which is 
$1$. Using \eqref{eq:arotate}, conclude that $A^\prime(\ho)=A(\ho^\prime)$.
\end{proof}
Since the operators $A^\prime$ and $A$ are unitarily equivalent, hereafter we identify 
$A^\prime$ (resp. $W^\prime$) with $A$ (resp. $W$).
\section{Field of orthonormal bases}\label{sec:bases}
Here we construct the field of orthonormal bases of $\K(k)$ that transform under 
rotations as an irreducible tensor operator.

By construction, $T(V(k))=\op_\lambda T(V_\lambda(k))$ as a vector space. The ideal 
$I(V(k))$ is homogeneous in the sense that it obeys the form 
$\op_\lambda I(V_\lambda(k))$, where $I(V_\lambda(k))$ denotes
$I(V(k))\mcap T(V_\lambda(k))$. Thus 
\begin{equation}
S(V(k))=\op_\lambda S(V_\lambda(k)), \quad  
S(V_\lambda(k))\dfn T(V_\lambda(k))/I(V_\lambda(k)).
\label{eq:SVlambda}
\end{equation}
We explore this fact below.

For $\iota=((\lambda,\rho_1),\ldots,(\lambda,\rho_n))$,
$\rho=(\rho_1,\ldots,\rho_n)\in\indxl^n$, put 
\[\he_{\lambda^n\rho}(k)\dfn\he_\iota(k).\]
When $n=0$, $\rho=\rho_0$ is empty and $\he_{\lambda^0\rho_0}(k)\dfn\e(k)$. Thus
\begin{equation}
(\he_{\lambda^n\rho}(k)\vrt \rho\in\indxl^n)_{n\in\N_0}
\label{eq:lambdabasis}
\end{equation}
is the basis of $S(V_\lambda(k))$.

Further, for $\iota^\prime=((\lambda,\rho_1^\prime),
\ldots,(\lambda,\rho_{n^\prime}^\prime))$, $\rho^\prime=(\rho_1^\prime,\ldots,
\rho_{n^\prime}^\prime)\in\indxl^{n^\prime}$, put 
\[\he_{\lambda^{n+n^\prime}\rho\rho^\prime}(k)\dfn\he_{\iota\iota^\prime}(k).\]
The action of $\Sym_n$ on $\rho\in\indxl^n$ is defined similar to 
that of $\Sym_n$ on $\iota\in\indx^n$.
\subsection{SCFP}\label{sec:scfps}
Let $\Gamma_{\lambda^n}$ be the set of pairs
\[\gamma\dfn(\Lambda,M), \quad 
M\in\indx_\Lambda\dfn\{-\Lambda,-\Lambda+1,\ldots,\Lambda\}\] 
where $\Lambda$ is obtained by reducing the tensor product of $n$ copies of 
$SO_3$-irreducible representations $[\lambda]$. We look for a transformation---as a 
collection of coefficients $c_\rho(\lambda^n\gamma)\in\C$---that maps a measurable field
$k\lmap v_{\lambda^n\gamma}(k)$ of orthonormal bases onto the field of bases in
\eqref{eq:lambdabasis}:
\begin{equation}
\he_{\lambda^n\rho}(k)=\sum_{\gamma\in \Gamma_{\lambda^n}} 
c_\rho(\lambda^n\gamma)v_{\lambda^n\gamma}(k).
\label{eq:c}
\end{equation}
By \eqref{eq:c}, $c_\rho(\lambda^n\gamma)$ is invariant under the action of $\Sym_n$.
When $n=0$ and $n=1$, \eqref{eq:c} is trivial
\begin{subequations}\label{eq:cv01}
\begin{align}
c_{\rho_0}(\lambda^0\Lambda M)\dfn&\delta_{\Lambda 0}\delta_{M0}, \quad  
v_{\lambda^0\Lambda M}(k)\dfn\delta_{\Lambda 0}\delta_{M0}\e(k), \label{eq:cv01-a} \\
c_\rho(\lambda^1\Lambda M)\dfn&\delta_{\Lambda \lambda}\delta_{M\rho}, \quad 
v_{\lambda^1\Lambda M}(k)\dfn\delta_{\Lambda \lambda}e_{\lambda M}(k).
\end{align}
\end{subequations}
We require $(v_{\lambda^n\gamma}(k))$ be an orthonormal basis of $\K(k)$. By
\eqref{eq:cv01-a}, $v_{\lambda^0\gamma}(k)$ is an element of $T^0(V(k))$ and it does not 
depend on $\lambda$. Thus we have 
\begin{subequations}\label{eq:vorth}
\begin{align}
\braket{v_{\lambda^n\gamma}(k),v_{\lambda^{\prime\,n^\prime}\gamma^\prime}(k)}_{\K(k)}
=&\delta_{nn^\prime}\delta_{\gamma\gamma^\prime}\delta_{\lambda\lambda^\prime},
\quad n,n^\prime\in\N, \label{eq:vorth-a} \\
\braket{v_{\lambda^n\gamma}(k),\e(k)}_{\K(k)}
=&\delta_{n0}\delta_{\gamma0}, \quad n\in\N_0. \label{eq:vorth-b}
\end{align}
\end{subequations}
Here and elsewhere, $\delta_{\gamma\gamma^\prime}$ reads $\delta_{\Lambda\Lambda^\prime}
\delta_{MM^\prime}$, and $\delta_{\gamma0}$ reads $\delta_{\Lambda0}\delta_{M0}$.

Now assume that $n=2$; $\rho=(\rho_1,\rho_2)\in\indxl^2$. We have that $[\lambda]\ot
[\lambda]$ reduces to $[\Lambda]$ for $\Lambda\in\{0,1,\ldots,2\lambda\}$. However,
using the first commutation relation in \eqref{eq:ccr}, 
the basis vector $\{\he_{\lambda^2}(k)\}_{\Lambda M}$ of the space of
the representation labeled by $[\Lambda]$ vanishes identically for $\Lambda$ odd.
That is
\[\{e_\lambda(k)\hot e_{\lambda^\prime}(k)\}_{\Lambda M}=
(-1)^{\lambda+\lambda^\prime-\Lambda}
\{e_{\lambda^\prime}(k)\hot e_\lambda(k)\}_{\Lambda M}.\]
Thus, for $\lambda=\lambda^\prime$, $\Lambda\in\{0,2,\ldots,
2\lambda\}$. Using the latter, \eqref{eq:c} holds for  
\begin{subequations}\label{eq:cm1m2}
\begin{align}
c_{\rho_1\rho_2}(\lambda^2\Lambda M)\dfn&(\lambda\lambda\Vert\lambda^2\Lambda)
\cgc{\lambda}{\lambda}{\Lambda}{\rho_1}{\rho_2}{M}, \label{eq:cm1m2-a}\\
v_{\lambda^2\Lambda M}(k)\dfn&(\lambda\lambda\Vert\lambda^2\Lambda)
\{\he_{\lambda^2}(k)\}_{\Lambda M}, \quad
(\lambda\lambda\Vert\lambda^2\Lambda)\dfn\frac{1+(-1)^\Lambda}{2}.
\label{eq:cm1m2-b}
\end{align}
\end{subequations}

A one-to-one correspondence between $(v_{\lambda^n\gamma}(k)\vrt \gamma\in
\Gamma_{\lambda^n})_{n\in\N_0}$ and \eqref{eq:lambdabasis} 
requires the transformation in \eqref{eq:c} to have an inverse. We look for the inverse
transformation by generalizing \eqref{eq:cm1m2-b} by induction
\begin{equation}
v_{\lambda^n\Lambda M}(k)=\sum_{\Lambda^\prime}
\{v_{\lambda^{n-1}\Lambda^\prime}(k)\hot e_{\lambda}(k)\}_{\Lambda M}
(\lambda^{n-1}(\Lambda^\prime)\lambda\Vert\lambda^n\Lambda)
\label{eq:vjn}
\end{equation}
and $\Lambda^\prime$ is obtained by reducing the tensor product of $n-1$ copies of 
$SO_3$-irreducible representations $[\lambda]$. One assumes that the coefficient
satisfies
\begin{subequations}
\begin{align}
(\lambda^0(\Lambda^\prime)\lambda\Vert\lambda^1\Lambda)\dfn&
\delta_{\Lambda^\prime0}\delta_{\Lambda\lambda}, \\
(\lambda^1(\Lambda^\prime)\lambda\Vert\lambda^2\Lambda)\dfn&
\delta_{\Lambda^\prime\lambda}(\lambda\lambda\Vert\lambda^2\Lambda), \\
(\lambda^{n-1}(\Lambda^\prime)\lambda\Vert\lambda^n\Lambda)\equiv& 
\ol{(\lambda^n\Lambda\Vert\lambda^{n-1}(\Lambda^\prime)\lambda )}, \quad n\in\N.
\end{align}
\end{subequations}
Then, relation \eqref{eq:vjn} holds for all $n\in\N$, and it is compatible with
\eqref{eq:cv01} and \eqref{eq:cm1m2-b}.

The coefficient formally plays the same role as the coefficient of fractional parentage 
(CFP) in an antisymmetric case of fermionic particles (when describing electron-shells of 
an atom). Thus we call $(\lambda^{n-1}(\Lambda^\prime)\lambda\Vert\lambda^n\Lambda)$ the 
\textit{symmetric coefficient of fractional parentage} (SCFP).
\subsection{Coefficient identities}\label{sec:scfps-prop}
The properties of the SCFPs are somewhat similar to those of regular CFPs. For our
purposes we need only two of them. By \eqref{eq:vorth}, \eqref{eq:vjn}
\begin{equation}
1=\sum_{\Lambda^\prime}
\abs{(\lambda^{n-1}(\Lambda^\prime)\lambda\Vert \lambda^n\Lambda)}^2, \quad n\in\N
\label{eq:cfp-orth}
\end{equation}
(\cf \cite[Eq.~(9.13)]{Rudzikas}, \cite[Eq.~(11.9)]{Jucys-c}). Applying relation 
\eqref{eq:vjn} twice
\begin{align}
0=&\sum_{\Lambda_1}
\sqrt{2\Lambda_1+1}
\sixj{\lambda}{\lambda}{\Lambda^\prime}{\Lambda_2}{\Lambda}{\Lambda_1} \nonumber \\
&\cdot(\lambda^{n-2}(\Lambda_2)\lambda\Vert\lambda^{n-1}\Lambda_1)
(\lambda^{n-1}(\Lambda_1)\lambda\Vert\lambda^n\Lambda), \quad n\in\N_{\geq2}
\label{eq:recSCFP}
\end{align}
for $\Lambda^\prime$ odd (\cf \cite[Eq.~(9.12)]{Rudzikas}, \cite[Eq.~(11.8)]{Jucys-c}).
The coefficient $\{\cdots\}$ is the $6j$-symbol. Using \eqref{eq:cfp-orth} and
\eqref{eq:recSCFP}, and assuming that the SCFPs are real numbers, some numerical
values are shown in Tab.~\ref{tab:t1}. There, the SCFPs with missing numbers $\Lambda$, 
$\Lambda^\prime$ allowed by the rules of angular reduction vanish identically.
\begin{table}[h!]
\centering
\caption{Numerical values of some 
$(\lambda^{n-1}(\Lambda^\prime)\lambda\Vert\lambda^n\Lambda)$.}\label{tab:t1}
\begin{tabular}{ c  c  c | c  c  c 
}
$\lambda^{n-1}(\Lambda^\prime)$ & $\Lambda$ & 
$(\lambda^{n-1}(\Lambda^\prime)\lambda\Vert\lambda^n\Lambda)$ &
$\lambda^{n-1}(\Lambda^\prime)$ & $\Lambda$ & 
$(\lambda^{n-1}(\Lambda^\prime)\lambda\Vert\lambda^n\Lambda)$ \\ 
\hline\hline \\ [-2.ex]
$1^2(0)$ & $1$ & $\sqrt{5/3}$ & 
$2^2(2)$ & $4$ & $\sqrt{11/21}$
\\ [.5ex]
$1^2(2)$ & $1$ & $2/3$ & 
$2^2(4)$ & $2$ & $2\sqrt{3/35}$
\\ [.5ex]
$2^2(0)$ & $2$ & $\sqrt{7/15}$ &
& $3$ & $-\sqrt{2/7}$
\\ [.5ex]
$2^2(2)$ & $2$ & $\sqrt{2/21}$ &
& $4$ & $\sqrt{10/21}$
\\ [.5ex]
 & $3$ & $\sqrt{5/7}$ &
&&
\\ [.5ex]
\hline\hline
\end{tabular}
\end{table}

By analogy to the recurrence relation \eqref{eq:recSCFP} for the SCFPs, one finds the 
recurrence relation for the coefficients in \eqref{eq:c}.
\begin{lem}\label{lem:rec1}
Let $\rho_1\in\indxl$ and $\rho\in\indxl^n$ and
$(\Lambda,M)\in\Gamma_{\lambda^{n+1}}$ and $n\in\N_0$. Then 
\[c_{\rho\rho_1}(\lambda^{n+1}\Lambda M)=\sum_{(\Lambda^\prime,M^\prime)\in
\Gamma_{\lambda^n}}
c_{\rho}(\lambda^n\Lambda^\prime M^\prime)
(\lambda^{n+1}\Lambda\Vert\lambda^n(\Lambda^\prime)\lambda) 
\cgc{\Lambda^\prime}{\lambda}{\Lambda}{M^\prime}{\rho_1}{M}.\]
\end{lem}
\begin{proof}
By \eqref{eq:cv01} and \eqref{eq:cm1m2-a}, it suffices to examine the case $n\geq2$.
Using \eqref{eq:b} and then \eqref{eq:c}
\begin{equation}
b_{\lambda\rho_1}(k)\he_{\lambda^n\rho}(k) 
=\sqrt{n+1}\sum_{\gamma_1\in\Gamma_{\lambda^{n+1}}}
c_{\rho\rho_1}(\lambda^{n+1}\gamma_1)
v_{\lambda^{n+1}\gamma_1}(k). 
\label{eq:1}
\end{equation}
On the other hand, using \eqref{eq:c} and then \eqref{eq:b}
\begin{equation}
b_{\lambda\rho_1}(k)\he_{\lambda^n\rho}(k)
=\sqrt{n+1}\sum_{\gamma_2\in\Gamma_{\lambda^n}}
c_{\rho}(\lambda^n\gamma_2)
v_{\lambda^n\gamma_2}(k)\hot e_{\lambda\rho_1}(k). 
\label{eq:2}
\end{equation}
Projecting \eqref{eq:1} $=$ \eqref{eq:2} on $v_{\lambda^{n+1}\gamma_1}(k)$ and 
then using \eqref{eq:vorth-a}
\begin{align}
c_{\rho\rho_1}(\lambda^{n+1}\gamma_1)=&
\sum_{\gamma_2\in\Gamma_{\lambda^n}}\sum_{\gamma_3\in\Gamma_{\lambda^{n+1}}}
c_{\rho}(\lambda^n\gamma_2)
\cgc{\Lambda_2}{\lambda}{\Lambda_3}{M_2}{\rho_1}{M_3} \nonumber \\
&\cdot\braket{v_{\lambda^{n+1}\gamma_1}(k), 
\{v_{\lambda^n\Lambda_2}(k)\hot e_\lambda(k)\}_{\gamma_3}}_{\K(k)}. 
\label{eq:3}
\end{align}

Next, using \eqref{eq:vorth-a} and \eqref{eq:vjn}
\[\delta_{\gamma_1\gamma_3}=
\sum_{\Lambda_2}(\lambda^n(\Lambda_2)\lambda\Vert\lambda^{n+1}\Lambda_3)
\braket{v_{\lambda^{n+1}\gamma_1}(k), 
\{v_{\lambda^n\Lambda_2}(k)\hot e_\lambda(k)\}_{\gamma_3}}_{\K(k)}.\]
The latter together with \eqref{eq:cfp-orth} implies
\begin{equation}
\braket{v_{\lambda^{n+1}\gamma_1}(k),
\{v_{\lambda^n\Lambda_2}(k)\hot e_\lambda(k)\}_{\gamma_3}}_{\K(k)}
=\delta_{\gamma_1\gamma_3} 
(\lambda^{n+1}\Lambda_1\Vert\lambda^n(\Lambda_2)\lambda).
\label{eq:4}
\end{equation}
Substitute \eqref{eq:4} in \eqref{eq:3} and deduce the relation as claimed.
\end{proof}
\begin{cor}\label{cor:corth}
Let $n\in\N_0$.
For $\gamma\in\Gamma_{\lambda^n}$ such that $c_\rho(\lambda^n\gamma)$ is not 
identically zero 
\begin{equation}
\sum_{\rho\in\indxl^n}\ol{c_\rho(\lambda^n\gamma)}c_\rho(\lambda^n\gamma^\prime)
=\delta_{\gamma\gamma^\prime}.
\label{eq:jht1}
\end{equation}
\end{cor}
\begin{proof}
Let $\rho_1\in\indxl$ and $\rho\in\indxl^n$ and $\gamma_1,\gamma_2
\in\Gamma_{\lambda^{n+1}}$. By Lemma~\ref{lem:rec1}
\begin{align}
\sum_{\rho,\rho_1}&\ol{c_{\rho\rho_1}(\lambda^{n+1}\gamma_1)}
c_{\rho\rho_1}(\lambda^{n+1}\gamma_2) \nonumber \\
=&\sum_{\gamma_1^\prime,\gamma_2^\prime\in\Gamma_{\lambda^n}}
(\lambda^{n+1}\Lambda_2\Vert\lambda^n(\Lambda_2^\prime)\lambda)
(\lambda^n(\Lambda_1^\prime)\lambda\Vert\lambda^{n+1}\Lambda_1) \nonumber \\
&\cdot\sum_{\rho_1}
\cgc{\Lambda_1^\prime}{\lambda}{\Lambda_1}{M_1^\prime}{\rho_1}{M_1}
\cgc{\Lambda_2^\prime}{\lambda}{\Lambda_2}{M_2^\prime}{\rho_1}{M_2} 
\sum_\rho\ol{c_\rho(\lambda^n\gamma_1^\prime)}
c_\rho(\lambda^n\gamma_2^\prime). 
\label{eq:eerrtt}
\end{align}
We argue by induction.
Assume that the sum over $\rho$ on the right-hand side is
$\delta_{\gamma_1^\prime\gamma_2^\prime}$. Then, using \eqref{eq:cfp-orth}, the 
right-hand side of \eqref{eq:eerrtt} is $\delta_{\gamma_1\gamma_2}$. That is, if 
\eqref{eq:jht1} holds for $c_\rho(\lambda^n\gamma)$ with $\rho\in\indxl^n$,
$\gamma\in\Gamma_{\lambda^n}$, then \eqref{eq:jht1} holds for 
$c_\rho(\lambda^{n+1}\gamma)$ with $\rho\in\indxl^{n+1}$,
$\gamma\in\Gamma_{\lambda^{n+1}}$. Now, we know from \eqref{eq:cm1m2-a} that 
\eqref{eq:jht1} is valid for $n=2$. Thus, \eqref{eq:jht1} holds for all $n\geq2$. 
For $n=0$ and $n=1$, \eqref{eq:jht1} applies trivially; see \eqref{eq:cv01}.
\end{proof}
\subsection{Hilbert space decomposition}\label{sec:orth-bases}
With the help of the previously obtained results one deduces the following.
\begin{thm}\label{thm:corth2}
The sequence
\[(\e, 
(v_{\lambda^n\gamma}\vrt\gamma\in\Gamma_{\lambda^n})_{(\lambda,n)\in\N_0\times\N})
\quad \text{with} \quad
v_{\lambda^n\gamma}=
\sum_{\rho\in\indxl^n}\ol{c_{\rho}(\lambda^n\gamma)}\he_{\lambda^n\rho}\]
is a measurable field of orthonormal bases of $\K(k)$. 
\end{thm}
\begin{proof}
We show that the field of orthonormal bases
\begin{subequations}\label{eq:Bk}
\begin{align}
B_\lambda(k)=&(\e(k), 
(v_{\lambda^n\gamma}(k)\vrt\gamma\in\Gamma_{\lambda^n})_{n\in\N})
\quad \text{with} \label{eq:Bk-a} \\
v_{\lambda^n\gamma}(k)=&
\sum_{\rho\in\indxl^n}\ol{c_{\rho}(\lambda^n\gamma)}\he_{\lambda^n\rho}(k)
\quad \text{all $k$} \label{eq:Bk-b}
\end{align}
\end{subequations}
is in one-to-one correspondence with a measurable field of bases in 
\eqref{eq:lambdabasis}. Then the proof is accomplished by using \eqref{eq:SVlambda},
since $(B_\lambda(k))_{\lambda\in\N_0}$ is a sequence of orthonormal vectors
by \eqref{eq:vorth}.

Relation \eqref{eq:Bk-b} follows from \eqref{eq:c} and Corollary~\ref{cor:corth}.
We show the converse. Multiply \eqref{eq:Bk-b} by $c_{\rho}(\lambda^n\gamma)$ and
sum over $\gamma\in\Gamma_{\lambda^n}$. We have 
\[\sum_{\gamma}c_{\rho}(\lambda^n\gamma)v_{\lambda^n\gamma}(k)
=\sum_{\rho^\prime}C_{\lambda^n\rho\rho^\prime}
\he_{\lambda^n\rho^\prime}(k)\]
with
\[C_{\lambda^n\rho\rho^\prime}\dfn\sum_{\gamma}
c_{\rho}(\lambda^n\gamma)\ol{c_{\rho^\prime}(\lambda^n\gamma)}.\]
On the other hand, projecting \eqref{eq:Bk-b} on $\he_{\lambda^n\rho}(k)$ and then
using \eqref{eq:scalarSV} and that $c_{\rho}(\lambda^n\gamma)$ is invariant under
the action of $\Sym_n$
\[c_\rho(\lambda^n\gamma)=\braket{v_{\lambda^n\gamma}(k),
\he_{\lambda^n\rho}(k)}_{\K(k)}.\]
By assumption, $B_\lambda(k)$ is complete; hence
\begin{align*}
C_{\lambda^n\rho\rho^\prime}=&\sum_\gamma
\braket{\he_{\lambda^n\rho^\prime}(k),v_{\lambda^n\gamma}(k)}_{\K(k)}
\braket{v_{\lambda^n\gamma}(k),\he_{\lambda^n\rho}(k)}_{\K(k)} \\
=&\braket{\he_{\lambda^n\rho^\prime}(k),\he_{\lambda^n\rho}(k)}_{\K(k)}
=\frac{1}{n!}\sum_{\pi\in\Sym_n}\delta_{\rho^\prime,\pi\cdot\rho}.
\end{align*}
\ie 
\[\sum_{\rho^\prime}C_{\lambda^n\rho\rho^\prime}
\he_{\lambda^n\rho^\prime}(k)=\he_{\lambda^n\rho}(k).\]
Thus, \eqref{eq:Bk-a} is in one-to-one correspondence with \eqref{eq:lambdabasis}.
\end{proof}
The key conclusion following from Theorem~\ref{thm:corth2} is that $\K$ is the orthogonal 
sum $\op_\Gamma\K_\Gamma$ of irreducible invariant subspaces $\K_\Gamma$ spanned by 
square-integrable vector fields 
\[\{\e,v_{\Gamma M}\vrt M\in\indx_\Lambda\}, \quad 
\Gamma\equiv(\lambda^n,\Lambda), \quad n\in\N.\] 
Further, using the rules for reducing the tensor product $[J]\ot[\Lambda]\lto[L]$, the 
space $\L_{0,J}\ot\K_\Gamma$ is the orthogonal sum $\op_L\H_{J\Gamma L}$ of irreducible 
invariant subspaces spanned by square-integrable vector fields 
\begin{equation}
\{w_{LM_L}\ot\e,h_{J\Gamma LM_L}\vrt M_L\in\indx_L\}, \quad
h_{J\Gamma LM_L}\dfn\{w_J\ot v_\Gamma\}_{LM_L}.
\label{eq:hbasis}
\end{equation}
Therefore
\begin{subequations}\label{eq:Hdecom}
\begin{equation}
\H=\H^0\op\Hex
\end{equation}
where 
\begin{align}
\H^0\dfn&\bigoplus_L\H_L^0, \quad
\H_L^0\dfn\L_{0,L}\ot\C\e, \\
\Hex\dfn&\bigoplus_L\HLex, \quad 
\HLex\dfn
\bigoplus_{\substack{J,\Gamma \\ (n\geq1)}} \H_{J\Gamma L}.
\end{align}
Physically, $\H^0$ is the Hilbert space of vector fields without phonons ($n=0$); 
$\Hex$ is the Hilbert space of vector fields that account for phonon 
excitations $(n\geq1)$. 

The spaces $\H_L^0$ and $\HLex$ are direct integrals of measurable fields
$k\lmap\L_{0,L}\ot\C\e(k)$ and $k\lmap \HLex(k)$, respectively, of Hilbert spaces over
$(\Rp,\F,\mu)$. The space $\H_L^0(k)$ has an orthonormal basis 
$(w_{LM_L}\ot\e(k))_{M_L}$; $\HLex(k)$ is the
orthogonal sum of Hilbert spaces $\H_{J\Gamma L}(k)$ with orthonormal bases 
$(h_{J\Gamma LM_L}(k))_{M_L}$.

We have that
\begin{equation}
\H=\bigoplus_L\H_L, \quad \H_L\dfn\H_L^0\op\HLex.
\end{equation}
\end{subequations}
Using \eqref{eq:Hdecom}, an element $\psi_{LM_L}$ of $\H_L$ is a vector
field $k\lmap\psi_{LM_L}(k)$, with the value
\begin{subequations}\label{eq:psiLabc}
\begin{equation}
\psi_{LM_L}(k)=c_{LM_L}(k)w_{LM_L}\ot\e(k)+\sum_{J,\Gamma}
c_{J\Gamma LM_L}(k)h_{J\Gamma LM_L}(k)
\label{eq:psiL}
\end{equation}
belonging to $\H_L(k)$ for $\mu$-a.e. $k$. The sum over $(J,\Gamma)\in Z$ (an index set)
is a vector norm limit of partial sums. The coordinates
\begin{align}
c_{LM_L}(k)\equiv &c_{L\lambda^00LM_L}(k)\dfn
\braket{w_{LM_L}\ot\e(k),\psi_{LM_L}(k)}_{\K(k)}, \label{eq:c0} \\
c_{J\lambda^n\Lambda LM_L}(k)\dfn&
\braket{h_{J\lambda^n\Lambda LM_L}(k),\psi_{LM_L}(k)}_{\K(k)}, \quad n\in\N
\label{eq:cn}
\end{align}
are such that
\begin{align}
&k\lmap(c_{LM_L}(k))_{M_L\in\indx_L} \;\text{is of class}\;
L^2(\Rp,\mu;\ell^2(\indx_L)), \\
&k\lmap (c_{J\Gamma LM_L}(k))_{((J,\Gamma),M_L)\in Z\times\indx_L}
\;\text{is of class}\; \nonumber \\
&\qquad\qquad\qquad\qquad\qquad\qquad\quad L^2(\Rp,\mu;\ell^2(Z\times\indx_L)).
\end{align}
\end{subequations}
According to \eqref{eq:c00}, there exists a finite $N\in\N_0$ such that
$c_{J\Gamma L M_L}(k)$ vanishes identically for $n>N$.
Thus, the sum over $n\in\N$ in \eqref{eq:psiL} is actually the sum over
$n\in\{1,2,\ldots,N\}$ for some finite $N$.

An element $\psi$ of $\H$ is a vector field $k\lmap\psi(k)$, with the value
given by the linear span of \eqref{eq:psiL}.
\section{Decomposition of angulon operator}\label{sec:reduce}
Here we calculate the matrix elements of the (rotated) angulon operator 
\eqref{eq:a1prime} with respect to the field \eqref{eq:hbasis} of orthonormal bases of
$\H_L^0(k)\op\H_{J\Gamma L}(k)$; here and elsewhere, we assume that
$n\in\N$ when we write $\H_{J\Gamma L}(k)$. The results have direct influence on the 
spectral analysis of angulon.

The orthonormal bases represent the $SO_3$-irreducible tensor operator of rank $L$.
Since the angulon operator $A$ is the $SO_3$-scalar tensor operator, the
Wigner--Eckart theorem \cite[Eq.~(5.15)]{Rudzikas}, \cite[Eq.~(3.41)]{Rudz-Kan}, 
\cite[Sec.~I.2]{Jucys-c} implies that the matrix elements of $A$ are diagonal with
respect to $(L,M_L)$. Thus $A$ maps $\dom A\mcap\H_L$ into $\H_L$, \ie
$\H_L$ is an invariant subspace for $A$. Then, the operator
\begin{equation}
A_L\dfn A\vert\dom A_L, \quad \dom A_L\dfn\dom A\mcap\H_L
\label{eq:AL}
\end{equation}
is the part of $A$ in $\H_L$. The domain $\dom A_L$ consists of measurable vector
fields $k\lmap\psi_{LM_L}(k)$, with $\psi_{LM_L}(k)$ as in \eqref{eq:psiLabc}, and with 
the coordinates satisfying in addition the property that
\begin{subequations}\label{eq:coord}
\begin{align}
&k\lmap(L(L+1)c_{LM_L}(k))_{M_L\in\indx_L} \;\text{is of class}\;
L^2(\Rp,\mu;\ell^2(\indx_L)), \\
&k\lmap (J(J+1)c_{J\Gamma LM_L}(k))_{((J,\Gamma),M_L)\in Z\times\indx_L}
\;\text{is of class}\; \nonumber \\
&\qquad\qquad\qquad\qquad\qquad\qquad\qquad\qquad\quad 
L^2(\Rp,\mu;\ell^2(Z\times\indx_L)).
\end{align}
\end{subequations}
Thus $\dom A_L\subseteq \H_L$ densely.

The results below imply that $A_L$ is self-adjoint and decomposable. 
Indeed, let $P_L(k)$ be the projection of $\H(k)$ onto $\H_L(k)$. 
Let $B_L(k)\dfn P_L(k)A(k)$. Since $P_L(k)$ is orthogonal and $A(k)$ is self-adjoint, we 
have $B_L(k)^*=A(k)P_L(k)$. By Lemma~\ref{lem:matrix} below, $B_L(k)$ is symmetric. 
Thus $B_L(k)\subseteq B_L(k)^*$. On the other hand, 
$\dom B_L(k)^*\subseteq\H_L(k)$; so 
for $\psi_L(k)\in\dom B_L(k)^*$ and $\phi_L(k)\in\H_L(k)$ 
\begin{align*}
B_L(k)P_L(k)\psi_L(k)=&B_L(k)\psi_L(k), \\
\braket{B_L(k)P_L(k)\psi_L(k),\phi_L(k)}_{\H(k)}=&
\braket{A(k)P_L(k)\psi_L(k),\phi_L(k)}_{\H(k)}
\end{align*}
Thus $B_L(k)^*\subseteq B_L(k)$ and we have that $B_L(k)$ is a self-adjoint operator in 
$\H_L(k)$. Since the projection $P_L(k)$ commutes with $A(k)$, $\H_L(k)$ is a reducing 
(and hence invariant) 
subspace for $A(k)$ and we have $B_L(k)=A_L(k)$. To a measurable field $k\lmap A_L(k)$ of 
self-adjoint operators corresponds a self-adjoint decomposable operator 
$A_L=\int^{\op}A_L(k)\mu(dk)$. Using \eqref{eq:Hdecom}, $A$ is the orthogonal 
sum 
\begin{equation}
A=\bigoplus_L A_L
\label{eq:AL2}
\end{equation}
of self-adjoint decomposable operators \eqref{eq:AL}, and the sum over $L$ is
understood as a strong limit of partial sums.

Let us define
\begin{equation}
U_{\lambda L}(J^\prime\Lambda^\prime J\Lambda k)\dfn
(-1)^{J+L}U_\lambda(k)(J^\prime\Vert Y_\lambda \Vert J)
\sixj{J^\prime}{\Lambda^\prime}{L}{\Lambda}{J}{\lambda}
\label{eq:UUU}
\end{equation}
and
\begin{subequations}\label{eq:tauupsilon}
\begin{align}
\tau(\lambda^n\Lambda^\prime\Lambda)\dfn&(-1)^{\Lambda^\prime}
\sqrt{(n+1)(2\Lambda^\prime+1)}(\lambda^{n+1}\Lambda^\prime \Vert 
\lambda^n(\Lambda)\lambda), \quad n\in\N_0, \\
\upsilon(\lambda^n\Lambda^\prime\Lambda)\dfn&
\ol{\tau(\lambda^{n-1}\Lambda\Lambda^\prime)}, \quad n\in\N.
\end{align}
\end{subequations}
Here $(J^\prime\Vert Y_\lambda \Vert J)$ is the \textit{reduced} matrix element for
$\braket{w_{J^\prime M^\prime},Y_{\lambda\rho}w_{JM}}_{\L_0}$.
\begin{rems}\label{rems:rems1}
(1) $(J^\prime\Vert Y_\lambda\Vert J)=(-1)^{J^\prime-J}
\ol{(J\Vert Y_\lambda\Vert J^\prime)}$, which follows from the Wigner--Eckart theorem. 
For example, let $\L_0=\mathbb{S}^2$ be a unit sphere and $w_J=Y_J$; then 
\begin{equation}
(J^\prime\Vert Y_\lambda\Vert J)=\sqrt{\frac{(2\lambda+1)(2J+1)}{4\pi}}
\cgc{J}{\lambda}{J^\prime}{0}{0}{0}.
\label{eq:Ysubm}
\end{equation}
The reduced matrix element in \eqref{eq:Ysubm} is nonzero iff the integer 
$J+\lambda+J^\prime$ is even.

(2) $U_{\lambda L}(J^\prime\Lambda^\prime J\Lambda k)=
\ol{U_{\lambda L}(J\Lambda J^\prime\Lambda^\prime k)}$, which is due to 
\eqref{eq:UUU} and the above remark.
\end{rems}
\begin{lem}\label{lem:matrix}
The matrix elements of the angulon operator \eqref{eq:a1prime} with respect to the field 
\eqref{eq:hbasis} of orthonormal bases of $\H_L(k)$ are given by 
\begin{subequations}\label{eq:matrices}
\begin{align}
[A(k)]_{LM_L,LM_L}=&cL(L+1), \label{eq:matrices-a} \\
[A(k)]_{LM_L,J\lambda^n\Lambda LM_L}=&
\delta_{n1}\delta_{\Lambda\lambda}(-1)^\lambda\sqrt{2\lambda+1}
U_{\lambda L}(L0J\lambda k) \label{eq:matrices-b}
\intertext{for $n\in\N$, and}
[A(k)]_{J^\prime\lambda^{\prime\,n^\prime}\Lambda^\prime LM_L,J\lambda^n\Lambda LM_L}=&
\delta_{J^\prime J}\delta_{\Gamma^\prime\Gamma}(cJ(J+1)+n\omega(k)) \nonumber \\
&+\delta_{\lambda^\prime\lambda}U_{\lambda L}(J^\prime\Lambda^\prime J\Lambda k)
\nonumber \\
&\cdot(\delta_{n^\prime,n+1}\tau(\lambda^n\Lambda^\prime\Lambda)+
\delta_{n,n^\prime+1}\upsilon(\lambda^n\Lambda^\prime\Lambda)) \label{eq:matrices-c}
\end{align}
\end{subequations}
for $n,n^\prime\in\N$.
\end{lem}
\begin{proof}
Using \eqref{eq:J2} and \eqref{eq:n},
$\bJ^2 w_{LM_L}=L(L+1)w_{LM_L}$ and $n(k)\e(k)=0$ (all $k$). Using
\eqref{eq:b} and \eqref{eq:b+}, \eqref{eq:btilde}, $[W(k)]_{LM_L,LM_L}=0$ (all $k$);
hence \eqref{eq:matrices-a}.

By the same arguments, $[A^0(k)]_{LM_L,J\Gamma LM_L}=0$ (all $k$; $n\geq1$).
Using in addition \cite[Eq.~(3.47)]{Rudz-Kan}, for
$n\geq1$
\[[W(k)]_{LM_L,J\Gamma LM_L}=\frac{U_\Lambda(k)(L\Vert Y_\Lambda\Vert J)}{
\sqrt{(2\Lambda+1)(2L+1)}}(0 k\Vert b_\Lambda(k)^{\sim}\Vert \Gamma k).\]
On the one hand, the Wigner--Eckart theorem reads
\[\braket{1(k),b_{\Lambda,-M}(k)^{\sim}v_{\Gamma M}(k)}_{\K(k)}
=\frac{(-1)^{\Lambda+M}}{\sqrt{2\Lambda+1}}
(0 k\Vert b_\Lambda(k)^{\sim}\Vert \Gamma k).\]
On the other hand, by the definition of the annihilation operator
\[\braket{1(k),b_{\Lambda,-M}(k)^{\sim}v_{\Gamma M}(k)}_{\K(k)}
=\delta_{n1}\delta_{\Lambda\lambda}(-1)^{\lambda+M}.\]
This shows \eqref{eq:matrices-b} written in terms of \eqref{eq:UUU}.

Yet again by the same arguments
\begin{align*}
[A^0(k)]_{J^\prime\Gamma^\prime LM_L,J\Gamma LM_L}=&
\delta_{J^\prime J}\delta_{\Gamma^\prime\Gamma}(cJ(J+1)+n\omega(k)), \\
[W(k)]_{J^\prime\Gamma^\prime LM_L,J\Gamma LM_L}=&
(-1)^{\Lambda^\prime}\sum_{\lambda^{\prime\prime}}
U_{\lambda^{\prime\prime}L}(J^\prime\Lambda^\prime J\Lambda k) \\
&\cdot((\Gamma^\prime k\Vert b_{\lambda^{\prime\prime}}(k)\Vert\Gamma k)
+(-1)^{\lambda^{\prime\prime}}
(\Gamma^\prime k\Vert b_{\lambda^{\prime\prime}}(k)^{\sim}\Vert\Gamma k)).
\end{align*}
We claim that, for $n\geq1$, the reduced matrix element
\begin{equation}
(\Gamma^\prime k\Vert b_{\lambda^{\prime\prime}}(k)\Vert\Gamma k)
=\delta_{\lambda^\prime\lambda}\delta_{\lambda^{\prime\prime}\lambda}
\delta_{n^\prime,n+1}\sqrt{n+1}(\lambda^{n+1}\Lambda^\prime\Vert\lambda^n(\Lambda)
\lambda).
\label{eq:breduced}
\end{equation}
Then, using
\[(\lambda^{n-1}\Lambda^\prime k\Vert b_{\lambda}(k)^{\sim}\Vert
\lambda^n\Lambda k)=(-1)^{\Lambda+\Lambda^\prime+\lambda}
\ol{(\lambda^n\Lambda k\Vert b_\lambda(k)\Vert \lambda^{n-1}\Lambda^\prime k)}\]
one deduces \eqref{eq:matrices-c}.

By the Wigner--Eckart theorem, for $\gamma\in\Gamma_{\lambda^n}$ and
$\gamma^\prime\in\Gamma_{\lambda^{\prime\,n^\prime}}$ and
$\rho^{\prime\prime}\in\indx_{\lambda^{\prime\prime}}$
\[\braket{v_{\lambda^{\prime\,n^\prime}\gamma^\prime}(k),
b_{\lambda^{\prime\prime}\rho^{\prime\prime}}(k)v_{\lambda^n\gamma}(k)}_{\K(k)}
=\cgc{\Lambda}{\lambda^{\prime\prime}}{\Lambda^\prime}{M}{\rho^{\prime\prime}}{
M^\prime}\frac{(\Gamma^\prime k\Vert b_{\lambda^{\prime\prime}}(k)\Vert
\Gamma k)}{\sqrt{2\Lambda^\prime+1}}.\]
On the other hand, using \eqref{eq:b}
\begin{align*}
&\braket{v_{\lambda^{\prime\,n^\prime}\gamma^\prime}(k),
b_{\lambda^{\prime\prime}\rho^{\prime\prime}}(k)v_{\lambda^n\gamma}(k)}_{\K(k)}
=\delta_{n^\prime,n+1}\sqrt{n+1}\sum_{\Lambda^{\prime\prime},M^{\prime\prime}}
\cgc{\Lambda}{\lambda^{\prime\prime}}{\Lambda^{\prime\prime}}{M}{\rho^{\prime\prime}}{
M^{\prime\prime}} \\
&\qquad\qquad\qquad\qquad\qquad\qquad \cdot
\braket{v_{\lambda^{\prime\,n+1}\gamma^\prime}(k),
\{v_{\lambda^n\Lambda}(k)\hot e_{\lambda^{\prime\prime}}(k)\}_{\Lambda^{\prime\prime}
M^{\prime\prime}}}_{\K(k)}.
\end{align*}
By \eqref{eq:Bk-b} and then using \eqref{eq:4}, for $n\geq1$ 
\begin{align*}
&\braket{v_{\lambda^{\prime\,n+1}\gamma^\prime}(k),
\{v_{\lambda^n\Lambda}(k)\hot e_{\lambda^{\prime\prime}}(k)\}_{\Lambda^{\prime\prime}
M^{\prime\prime}}}_{\K(k)} \\
&\qquad =
\delta_{\lambda^\prime\lambda}\delta_{\lambda^{\prime\prime}\lambda}
\braket{v_{\lambda^{n+1}\gamma^\prime}(k),
\{v_{\lambda^n\Lambda}(k)\hot e_{\lambda}(k)\}_{\Lambda^{\prime\prime}
M^{\prime\prime}}}_{\K(k)} \\
&\qquad =
\delta_{\lambda^\prime\lambda}\delta_{\lambda^{\prime\prime}\lambda}
\delta_{\Lambda^{\prime\prime}\Lambda^\prime}
\delta_{M^{\prime\prime}M^\prime}(\lambda^{n+1}\Lambda^\prime\Vert
\lambda^n(\Lambda)\lambda)
\end{align*}
from which \eqref{eq:breduced} follows. This completes the proof of the lemma.
\end{proof}
\section{Eigenspace of angulon}\label{sec:eigen}
Let $E$ be an eigenvalue of the angulon operator $A$. According to the decomposition
\eqref{eq:AL2}, $E=E_{LM_L}$ is an eigenvalue of the part $A_L$ of $A$ in $\H_L$
for some $L$. Since $A_L$ is decomposable relative to $k\lmap\H_L(k)$,
$E$ is an eigenvalue of $A_L(k)$ for $\mu$-a.e. $k$.
\begin{thm}\label{thm:eigen}
Let $E=E_{LM_L}$ be the eigenvalue of $A_L$ for some $L$. For $\mu$-a.e. $k$,
the coordinates of the eigenvector satisfy \eqref{eq:psiLabc}, \eqref{eq:coord}, and
\begin{subequations}\label{eq:ceigen}
\begin{align}
0=&(cL(L+1)-E)c_{LM_L}(k) \nonumber \\
&+\sum_{J,\lambda}(-1)^\lambda\sqrt{2\lambda+1}
U_{\lambda L}(L0J\lambda k)c_{J\lambda^1\lambda LM_L}(k) \label{eq:ceigen-a} 
\intertext{and}
0=&\delta_{n1}\delta_{\Lambda\lambda}(-1)^\lambda\sqrt{2\lambda+1}
U_{\lambda L}(J\lambda L0k)c_{LM_L}(k) \nonumber \\
&+(cJ(J+1)+n\omega(k)-E)c_{J\lambda^n\Lambda LM_L}(k) \nonumber \\
&+\sum_{J^\prime,\Lambda^\prime}U_{\lambda L}(J\Lambda J^\prime\Lambda^\prime k)
[H(n-2)
\tau(\lambda^{n-1}\Lambda\Lambda^\prime)
c_{J^\prime\lambda^{n-1}\Lambda^\prime LM_L}(k) \nonumber \\
&+\upsilon(\lambda^{n+1}\Lambda\Lambda^\prime)
c_{J^\prime\lambda^{n+1}\Lambda^\prime LM_L}(k)]
\label{eq:ceigen-b}
\end{align}
\end{subequations}
for $n\in\N$. The step function $H(x)=1$ for $x\geq0$, and $H(x)=0$ for $x<0$.
\end{thm}
Let $\Omega_{LM_L}(k)$ be the set of solutions $(E_{LM_L},\psi_{LM_L}(k)\not\equiv0)$
satisfying \eqref{eq:psiLabc}, \eqref{eq:coord}, \eqref{eq:ceigen}.
$\Omega_{LM_L}(k)$ is defined for $\mu$-a.e. $k$ and is called
the \textit{eigenspace of $A_L$}.
\begin{rem}
Recall that $n\geq1$ in \eqref{eq:ceigen-b} is finite, \ie \eqref{eq:ceigen-b}
splits into $N$ equations for some finite $N$. When $n=N$, one puts
$c_{J^\prime\lambda^{N+1}\Lambda^\prime LM_L}(k)\equiv0$ in \eqref{eq:ceigen-b}.
\end{rem}
When $N=1$ we have Corollary~\ref{cor:ceigen}; when $N=2$ and $L=0$ we have 
Corollary~\ref{cor:ceigen2}.

Let us define
\begin{equation}
\Sigma_L^{(1)}(z,k)\dfn\sum_{J,\lambda}\frac{(2\lambda+1)
\abs{U_{\lambda L}(J\lambda L0k)}^2}{cJ(J+1)+\omega(k)-z}
\label{eq:selfk}
\end{equation}
for $\mu$-a.e. $k$ and some $z\in\C$; \ie if $O_\mu$ is the union of all $\mu$-null sets
then 
\begin{subequations}\label{eq:znull}
\begin{align}
z\in&\C\setm \sigma_1, \quad \text{where}\quad
\sigma_1\dfn\ol{\{\sigma_1(k)\vrt k\in\Rp\setm O_\mu\}}\quad\text{and} 
\label{eq:znull-1} \\
\sigma_1(k)\dfn&\ol{\{cJ(J+1)+\omega(k)\vrt J\in\N_0\}}\quad\text{for
$\mu$-a.e. $k$.} \label{eq:znull-2}
\end{align}
\end{subequations}
In particular, \eqref{eq:znull-1} holds for $z<0$. Note that $\sigma_1$ is the subset
of the essential spectrum of $A^0$. Indeed, $A^0$ is viewed as 
$A$ with $U_\lambda(k)\equiv0$ (all $\lambda$), and in this case 
Theorem~\ref{thm:eigen} says that $cJ(J+1)+n\omega(k)$ is an eigenvalue of infinite
multiplicity of $A^0(k)$ for $\mu$-a.e. $k$. Now, $z\in\sigma_1$ implies 
$k\in \{k\in\Rp\setm O_\mu\vrt z\in\sigma_1(k)\}$, which means that $z$ is such that
$\mu(\{k\in\Rp\vrt z\in\sigma_1(k)\})>0$. 

When $\omega(k)$ in \eqref{eq:znull} is replaced by
$2\omega(k)$, we write $\sigma_2$ instead of $\sigma_1$.
\begin{cor}\label{cor:ceigen}
For $N=1$, the eigenvalue $E=E_{LM_L}\in\R\setm \sigma_1$ of $A_L$ satisfies 
\begin{equation}
E=cL(L+1)-\Sigma_L^{(1)}(E,k)
\end{equation}
for $\mu$-a.e. $k$, and the coordinates satisfy  
\begin{equation}
c_{J\lambda^1\lambda LM_L}(k)=\frac{(-1)^{\lambda+1}\sqrt{2\lambda+1}
U_{\lambda L}(J\lambda L0k)}{cJ(J+1)+\omega(k)-E}c_{LM_L}(k)
\end{equation}
for $\mu$-a.e. $k$.
\end{cor}
Corollary~\ref{cor:ceigen} corresponds to the case when the many-body quantum state 
accounts for single bath excitations only \cite{Schm-Lem-a,Schm-Lem-b}. We see
that $E$ is of multiplicity $2L+1$. The coordinate $c_{LM_L}(k)$ is found from the
normalization condition.

When two-phonon excitations contribute notably, one solves \eqref{eq:ceigen} for
$N=2$. The simplest example is when $L=0$. To state our results we find it convenient 
to define
\begin{subequations}\label{eq:znull2}
\begin{align}
\sigma_*\dfn&\ol{\{\sigma_*(k)\vrt k\in\Rp\setm O_\mu\}}
\quad\text{with} \label{eq:znull2-1} \\
\sigma_*(k)\dfn&\ol{\bigcup_{\lambda\in\N_0}\{z\in\R\setm\sigma_2(k)\vrt
c\lambda(\lambda+1)+\omega(k)-z=\epsilon_\lambda(z,k)\}}
\end{align}
\end{subequations}
for $\mu$-a.e. $k$, where one puts
\begin{equation}
\epsilon_\lambda(z,k)\dfn2\frac{(-1)^\lambda U_\lambda(k)^2}{2\lambda+1}
\sum_\Lambda\frac{(\lambda\lambda\Vert\lambda^2\Lambda)
(\Lambda\Vert Y_\lambda\Vert\lambda)^2}{c\Lambda(\Lambda+1)+2\omega(k)-z}
\label{eq:vareps}
\end{equation}
for $\mu$-a.e. $k$ and $z\in\C\setm\sigma_2$. 

Let us further define
\begin{subequations}\label{eq:selfk2}
\begin{equation}
\Sigma_0^{(1,2)}(z,k)\dfn
\sum_\lambda 
\frac{U_\lambda(k)^2\abs{(\lambda\Vert Y_\lambda\Vert0)}^2}{
c\lambda(\lambda+1)+\omega(k)-z-\epsilon_\lambda(z,k)},
\quad z\in\C\setm\sigma_*
\label{eq:selfk2-1}
\end{equation}
and
\begin{align}
\Sigma_0^{(2)}(z,k)\dfn&\sum_\lambda
\frac{U_\lambda(k)^2\epsilon_\lambda(z,k)
\abs{(\lambda\Vert Y_\lambda\Vert0)}^2}{(c\lambda(\lambda+1)+\omega(k)-z)
(c\lambda(\lambda+1)+\omega(k)-z-\epsilon_\lambda(z,k))}, \nonumber \\
&\qquad\qquad\qquad\qquad\qquad\qquad\qquad\quad\quad
z\in\C\setm(\sigma_1\mcup\sigma_2\mcup\sigma_*)
\label{eq:selfk2-2}
\end{align}
\end{subequations}
for $\mu$-a.e. $k$.
\begin{rem}
Note that, when the reduced matrix element $(\Lambda\Vert Y_\lambda\Vert\lambda)$
is as in \eqref{eq:Ysubm}, \eqref{eq:vareps} formally coincides with 
\eqref{eq:Self-N2L0b}. We mean "formally", since the exact agreement between 
\eqref{eq:vareps} and \eqref{eq:Self-N2L0b} requires an additional assumption imposed
on the probability measure $\mu$. The same applies to \eqref{eq:selfk2-1} versus the
integrand in \eqref{eq:Self-N2L0}.
\end{rem}
With the above definitions, our result for $N=2$ is the following.
\begin{cor}\label{cor:ceigen2}
For $N=2$, the eigenvalue $E=E_{00}$ of $A_0$ satisfies 
\begin{subequations}\label{eq:E02}
\begin{align}
E=&-\Sigma_0^{(1,2)}(E,k), \quad E\in\R\setm\sigma_* \label{eq:E02-1} \\
=&-\Sigma_0^{(1)}(E,k)-\Sigma_0^{(2)}(E,k), \quad
E\in\R\setm(\sigma_1\mcup\sigma_2\mcup\sigma_*) \label{eq:E02-2}
\end{align}
\end{subequations}
for $\mu$-a.e. $k$, and the coordinates satisfy
\begin{subequations}\label{eq:c02}
\begin{align}
c_{\lambda\lambda^1\lambda00}(k)=&\frac{(-1)^{\lambda+1}U_\lambda(k)
(\lambda\Vert Y_\lambda\Vert0)}{
c\lambda(\lambda+1)+\omega(k)-E-\epsilon_\lambda(E,k)}c_{00}(k), \\
c_{\Lambda\lambda^2\Lambda00}(k)=&
\frac{(-1)^{\lambda+1}\sqrt{2}U_\lambda(k)(\lambda\lambda\Vert\lambda^2\Lambda)
(\Lambda\Vert Y_\lambda\Vert\lambda)}{\sqrt{2\lambda+1}
(c\Lambda(\Lambda+1)+2\omega(k)-E)}c_{\lambda\lambda^1\lambda00}(k)
\end{align}
\end{subequations}
for $E\in\R\setm(\sigma_2\mcup\sigma_*)$ and $\mu$-a.e. $k$.
\end{cor}
We close the present section by providing the proofs of the above statements.
\begin{proof}[Proof of Theorem~\ref{thm:eigen}]
Since $A_L$ is densely defined in $\H_L$, there exists a bounded decomposable
$B_L\in\B(\H_L)$ such that $A_L\subseteq B_L$. Hence we have $A_L(k)\subseteq B_L(k)$ for
$\mu$-a.e. $k$ (recall \eg \cite[Proposition~12.1.8(i)]{Schmudgen-2}), and the matrix
elements of $B_L$ with respect to the field \eqref{eq:hbasis} of orthonormal bases of 
$\H_L(k)$ are found from Lemma~\ref{lem:matrix}. Since $\B(\H_L)$ 
is the strong closure of the finite rank operators, we have
\[B_L(k)=\sum_{i,j}[B_L(k)]_{ij}\braket{f_j(k),\cdot}_{\H(k)}f_i(k)\]
($\mu$-a.e. $k$), where $[B_L(k)]_{ij}$ is the matrix element of $B_L(k)$ with respect to
the orthonormal basis $(f_i(k))$ of $\H_L(k)$; the infinite sum over indices $i,j$ is 
a strong limit of partial sums. Thus, for $(f_i)$ as in \eqref{eq:hbasis} 
and for $\psi_{LM_L}(k)\in\dom A_L(k)$
\begin{align}
A_L(k)\psi_{LM_L}(k)=&B_L(k)\psi_{LM_L}(k) \nonumber \\
=&\{[A(k)]_{LM_L,LM_L}\braket{w_{LM_L}\ot\e(k),\psi_{LM_L}(k)}_{\H(k)} \nonumber \\
&+\sum_{J,\Gamma}
[A(k)]_{LM_L,J\Gamma LM_L}\braket{h_{J\Gamma LM_L}(k),\psi_{LM_L}(k)}_{\H(k)}\} 
\nonumber \\
&\cdot w_{LM_L}\ot\e(k) \nonumber \\
&+\sum_{J,\Gamma}\{
[A(k)]_{J\Gamma LM_L,LM_L}\braket{w_{LM_L}\ot\e(k),\psi_{LM_L}(k)}_{\H(k)} 
\nonumber \\
&+\sum_{J^\prime,\Gamma^\prime}
[A(k)]_{J\Gamma LM_L,J^\prime\Gamma^\prime LM_L}
\braket{h_{J^\prime\Gamma^\prime LM_L}(k),\psi_{LM_L}(k)}_{\H(k)} \} \nonumber \\
&\cdot h_{J\Gamma LM_L}(k).
\label{eq:mx}
\end{align}
When $A_L(k)\psi_{LM_L}(k)=E_{LM_L}\psi_{LM_L}(k)$, using \eqref{eq:psiLabc} we thus have 
\eqref{eq:ceigen}.
\end{proof}
\begin{proof}[Proof of Corollary~\ref{cor:ceigen}]
Put $n=1$ in \eqref{eq:ceigen-b} and apply $c_{J\lambda^n\Lambda LM_L}(k)\equiv0$
for $n\geq2$ and $c_{J\lambda^1\Lambda LM_L}(k)=\delta_{\Lambda \lambda}
c_{J\lambda^1\lambda LM_L}(k)$ by \eqref{eq:cn}.
\end{proof}
\begin{proof}[Proof of Corollary~\ref{cor:ceigen2}]
First we note that, for $\mu$-a.e. $k$ and $z\in\C\setm(\sigma_1\mcup\sigma_2
\mcup\sigma_*)$, the relation
\[\Sigma_0^{(1,2)}(z,k)=\Sigma_0^{(1)}(z,k)+\Sigma_0^{(2)}(z,k)\]
follows directly from \eqref{eq:selfk} with $L=0$ and \eqref{eq:selfk2}.

When $L=0$, using $c_{J\lambda^1\Lambda LM_L}(k)=\delta_{\Lambda \lambda}
c_{J\lambda^1\lambda LM_L}(k)$ by \eqref{eq:cn},
relation~\eqref{eq:ceigen-b} with $n=1$ and $n=2$ reads
\begin{subequations}\label{eq:cn3}
\begin{align}
c_{\lambda\lambda^1\lambda00}(k)=&
\frac{(-1)^{\lambda+1}U_\lambda(k)(\lambda\Vert Y_\lambda\Vert0)}{
c\lambda(\lambda+1)+\omega(k)-E}c_{00}(k) \nonumber \\
&-\frac{\sqrt{2}U_\lambda(k)
\sum_\Lambda
(\lambda\lambda\Vert\lambda^2\Lambda)(\Lambda\Vert Y_\lambda\Vert \lambda)
}{\sqrt{2\lambda+1}
(c\lambda(\lambda+1)+\omega(k)-E)}
c_{\Lambda\lambda^2\Lambda00}(k) 
\intertext{for $E\in\R\setm\sigma_1$ and $\mu$-a.e. $k$, and}
c_{\Lambda\lambda^2\Lambda00}(k)=&
\frac{(-1)^{\lambda+1}\sqrt{2}U_\lambda(k)
(\lambda\lambda\Vert\lambda^2\Lambda)(\Lambda\Vert Y_\lambda\Vert \lambda)}{
\sqrt{2\lambda+1}(c\Lambda(\Lambda+1)+2\omega(k)-E)}c_{\lambda\lambda^1\lambda00}(k)
\nonumber \\
&+\frac{(-1)^{\lambda+\Lambda+1}U_\lambda(k)}{\sqrt{2\Lambda+1}
(c\Lambda(\Lambda+1)+2\omega(k)-E)} \nonumber \\
&\cdot\sum_{\Lambda^\prime}\frac{(\Lambda\Vert Y_\lambda\Vert\Lambda^\prime)
\upsilon(\lambda^3\Lambda\Lambda^\prime)}{2\Lambda^\prime+1}
c_{\Lambda^\prime\lambda^3\Lambda^\prime00}(k) 
\end{align}
\end{subequations}
for $E\in\R\setm\sigma_2$ and $\mu$-a.e. $k$. Since $N=2$, putting 
$c_{\Lambda^\prime\lambda^3\Lambda^\prime00}(k)\equiv0$ we deduce \eqref{eq:c02}.
Relations in \eqref{eq:E02} then follow from \eqref{eq:ceigen-a} and \eqref{eq:cn3},
since we have for $L=0$
\[Ec_{00}(k)=\sum_\lambda U_\lambda(k)(0\Vert Y_\lambda\Vert\lambda)
c_{\lambda\lambda^1\lambda00}(k)\]
for $\mu$-a.e. $k$.
\end{proof}
\section{Infimum of the spectrum with a.c. measure}\label{sec:num}
Let $\Theta_L$ be the numerical range of $A_L$. Then, $\Theta_L$ is the set 
\begin{equation}
\Theta_L\dfn\{\E[\psi_{LM_L}]\vrt \psi_{LM_L}\in\dom A_L\setm\{0\}\}
\label{eq:numL}
\end{equation}
of the real numbers
\begin{equation}
\E[\psi_{LM_L}]\dfn\norm{\psi_{LM_L}}_\H^{-2}
\braket{\psi_{LM_L},A_L\psi_{LM_L}}_\H, \quad \psi_{LM_L}\neq0.
\label{eq:Efunc}
\end{equation}
Note that $\psi_{LM_L}\neq0$ means that $\psi_{LM_L}(k)\neq0$ for $\mu$-a.e. $k$.
One regards $\E[\cdot]$ as a functional $\dom A_L\lto\R$ whose value is given by
\eqref{eq:Efunc}. We put
\begin{equation}
\E=\E_{LM_L}\dfn\inf\ol{\Theta_L}
\label{eq:infnumL}
\end{equation}
provided that $\E_{LM_L}>-\infty$ exists. By general principles, 
the spectrum $\sigma(A_L)$ of $A_L$ is contained in the closure $\ol{\Theta_L}$ and
\begin{equation}
\inf\sigma(A_L)=\E_{LM_L}.
\label{eq:infAL}
\end{equation}

Let us further define the functional $E[\cdot]\co\dom A_L\lto\R$ with the value
\begin{equation}
E[\psi_{LM_L}]\dfn\norm{\psi_{LM_L}}_\H^{-2}\int E_{LM_L}
\norm{\psi_{LM_L}(k)}_{\H(k)}^2\mu(dk)
\label{eq:thetaL-1}
\end{equation}
for $(E_{LM_L},\psi_{LM_L}(k))$ in the eigenspace $\Omega_{LM_L}(k)$. We denote by 
$\mathfrak{E}_L$ the set of all such \eqref{eq:thetaL-1}.
\begin{lem}\label{lem:infnumL}
$\E_{LM_L}=\inf\ol{\mathfrak{E}_L}$.
\end{lem}
\begin{proof}
Let $\psi_{LM_L}\in\dom A_L$. Using \eqref{eq:mx}, the minimization of
$\E[\psi_{LM_L}]$ with respect to $c_{LM_L}(k)$ and $c_{J\Gamma LM_L}(k)$ leads to 
\eqref{eq:ceigen-a} and \eqref{eq:ceigen-b}, respectively, with 
\[E_{LM_L}=\E[\psi_{LM_L}]\]
for $\mu$-a.e. $k$. Using \eqref{eq:numL}, 
\eqref{eq:infnumL}, and \eqref{eq:thetaL-1}, we thus have the result as claimed.
\end{proof}
\begin{rem}
For $\psi_{LM_L}$ as in \eqref{eq:thetaL-1}, Theorem~\ref{thm:eigen} says that 
$\psi_{LM_L}\neq0$ implies $c_{LM_L}(k)\neq0$ for $\mu$-a.e. $k$.
\end{rem}
The next Theorem~\ref{thm:infnumL} allows one to estimate $\E$ for a.c. $\mu$.
\begin{thm}\label{thm:infnumL}
Assume that $\mu\ll dk$, with the Radon--Nikodym derivative supported on $\ol{\Rp}$. 
Then
\begin{subequations}\label{eq:infeigen}
\begin{align}
\E_{LM_L}=&\min\{cL(L+1)-\Sigma_{LM_L},cL(L+1)\}, \quad \text{where} \\
\Sigma_{LM_L}\dfn&\sup\int\sum_{J,\lambda}(-1)^{\lambda+1}\sqrt{2\lambda+1}
U_{\lambda L}(L0J\lambda k)\frac{c_{J\lambda^1\lambda LM_L}(k)}{c_{LM_L}(k)}dk.
\end{align}
\end{subequations}
The supremum is taken over the points of the eigenspace $\Omega_{LM_L}(k)$ for a.e. $k$,
and the coordinates satisfy in addition
\begin{equation}
1=\abs{c_{LM_L}(k)}^2+\sum_{J,\Gamma}\abs{c_{J\Gamma LM_L}(k)}^2\quad\text{for a.e.
$k$.}
\label{eq:infcnorm}
\end{equation}
\end{thm}
The following Corollaries~\ref{cor:infnumL} and \ref{cor:infnumL2} are useful for
evaluating the contribution of single-phonon and two-phonon (for $L=0$) excitations.
\begin{cor}\label{cor:infnumL}
Assume that $\mu\ll dk$, with the Radon--Nikodym derivative supported on $\ol{\Rp}$. 
Then, $A_L\geq\E$, where the infimum $\E=\E_{LM_L}$ satisfies 
\begin{subequations}\label{eq:infeigen2}
\begin{align}
\E=&\min\{cL(L+1)-\Sigma_{L}^{(1)}(\E)-\Sigma_{LM_L}^\prime,cL(L+1)\}, 
\;\text{with} \\
\Sigma_{L}^{(1)}(\E)\dfn&\pv\Sigma_{L}^{(1)}(\E,k)dk, \quad
\text{$\Sigma_{L}^{(1)}(\cdot,k)$ as in \eqref{eq:selfk},} \label{eq:infeigen2-1} 
\intertext{provided that the principal value exists, and}
\Sigma_{LM_L}^\prime\dfn&\sup\pv\sum_{\lambda,\Lambda}
(-1)^\lambda\sqrt{2(2\lambda+1)(2\Lambda+1)}(\lambda\lambda\Vert\lambda^2\Lambda) \nonumber \\
&\cdot
\sum_J\frac{c_{J\lambda^2\Lambda LM_L}(k)}{c_{LM_L}(k)}\sum_{J^\prime}
\frac{U_{\lambda L}(L0J^\prime\lambda k)U_{\lambda L}(J^\prime\lambda J\Lambda k)}{
cJ^\prime(J^\prime+1)+\omega(k)-E}dk \label{eq:infeigen2-2}
\end{align}
\end{subequations}
where $E=E_{LM_L}$ is the eigenvalue of $A_L$.
The supremum is taken over the points of the eigenspace $\Omega_{LM_L}(k)$ for a.e. $k$,
with the coordinates satisfying in addition \eqref{eq:infcnorm}.
\end{cor}
\begin{cor}\label{cor:infnumL2}
Assume that $\mu\ll dk$, with the Radon--Nikodym derivative supported on $\ol{\Rp}$. 
Then, $A_0\geq\E$, where the infimum $\E=\E_{00}$ satisfies 
\begin{subequations}\label{eq:infeigen3}
\begin{align}
\E=&\min\{-\Sigma_{0}^{(1)}(\E)-\Sigma_{0}^{(2)}(\E)-\Sigma_{00}^{\prime\prime},0\}, 
\;\text{with} \\
\Sigma_{0}^{(j)}(\E)\dfn&\pv\Sigma_{0}^{(j)}(\E,k)dk,\quad j\in\{1,2\}
\end{align}
provided that the principal value exists, where $\Sigma_{0}^{(1)}(\cdot,k)$ and 
$\Sigma_{0}^{(2)}(\cdot,k)$ are as in \eqref{eq:selfk} and \eqref{eq:selfk2-2}, and
\begin{align}
\Sigma_{00}^{\prime\prime}\dfn&\sup\pv
\sum_\lambda\frac{(-1)^{\lambda+1}\sqrt{2}U_\lambda(k)^3(0\Vert Y_\lambda\Vert\lambda)}{
\sqrt{2\lambda+1}(c\lambda(\lambda+1)+\omega(k)-E-\epsilon_\lambda(E,k))} \nonumber \\
&\cdot\sum_\Lambda\frac{c_{\Lambda\lambda^3\Lambda00}(k)}{\sqrt{2\Lambda+1}c_{00}(k)}
\sum_{\Lambda^\prime}\frac{(\lambda\lambda\Vert\lambda^2\Lambda^\prime)
(\Lambda^\prime\Vert Y_\lambda\Vert\lambda)(\Lambda^\prime\Vert Y_\lambda\Vert\Lambda)
\upsilon(\lambda^3\Lambda^\prime\Lambda)}{\sqrt{2\Lambda^\prime+1}
(c\Lambda^\prime(\Lambda^\prime+1)+2\omega(k)-E)}dk
\end{align}
\end{subequations}
where $E=E_{00}$ is the eigenvalue of $A_0$.
The supremum is taken over the points of the eigenspace $\Omega_{00}(k)$ for a.e. $k$,
with the coordinates satisfying in addition \eqref{eq:infcnorm} with $L=0$.
\end{cor}
When single-phonon excitations contribute notably ($N=1$), 
Corollary~\ref{cor:infnumL} approximates to
\begin{equation}
\E=cL(L+1)-\Sigma_{L}^{(1)}(\E) \quad \text{for} \quad \Sigma_{L}^{(1)}(\E)\geq0.
\label{eq:SLEnergy}
\end{equation}
In particular, relation \eqref{eq:SLEnergy} holds true for $\E<0$. In this case
$\pv$ in \eqref{eq:infeigen2-1} is just $\int$. If we further assume
\eqref{eq:Ysubm}, we see that \eqref{eq:SLEnergy} is exactly \eqref{eq:E-SL},
\eqref{eq:Self-SL}. That is, Corollary~\ref{cor:infnumL} reduces to the equation for
the energy first obtained in \cite{Schm-Lem-a} by minimizing the energy functional based 
on an expansion in single bath excitations. The coincidence should not be considered as 
an unexpected result if one recalls the proof of Lemma~\ref{lem:infnumL}.

By contrast, when two-phonon excitations contribute notably ($N=2$),
Corollary~\ref{cor:infnumL2} approximates to
\begin{subequations}\label{eq:SLEnergy2}
\begin{align}
\E=&-\Sigma_{0}^{(1,2)}(\E) \\
=&-\Sigma_{0}^{(1)}(\E)-\Sigma_{0}^{(2)}(\E)
\end{align}
\end{subequations}
for $\Sigma_{0}^{(1,2)}(\E)\geq0$, where we put
\[\Sigma_{0}^{(1,2)}(\E)\dfn\pv\Sigma_{0}^{(1,2)}(\E,k)dk\]
with $\Sigma_{0}^{(1,2)}(\cdot,k)$ as in \eqref{eq:selfk2-1}. If we assume
\eqref{eq:Ysubm}, we see that $\Sigma_{0}^{(1,2)}(\E)$ is exactly \eqref{eq:Self-N2L0}.

We now close the present section by giving the proofs of Theorem~\ref{thm:infnumL}
and Corollaries~\ref{cor:infnumL} and \ref{cor:infnumL2}.
\begin{proof}[Proof of Theorem~\ref{thm:infnumL}]
Since $\mu(dk)=\phi(k)dk$, with $\supp\phi=\ol{\Rp}$ and $\phi>0$ a.e. on $\Rp$,
\[\abs{\sigma}=0\Longleftrightarrow\mu(\sigma)=0 \;\text{and hence}\;
\abs{\sigma}>0\Longleftrightarrow\mu(\sigma)>0\]
for $\sigma\in\F$. Here $\abs{\cdot}$ is a standard Lebesgue (length) measure. Thus,
the relations that hold for $\mu$-a.e. $k$, also hold for a.e. $k$, and vice verse.
In particular, $E_{LM_L}$ is an eigenvalue of $A_L$ iff $E_{LM_L}$ is an eigenvalue of
$A_L(k)$ for a.e. $k$. If, further, $\psi_{LM_L}(k)$ is the corresponding eigenvector of 
$A_L(k)$ such that $\norm{\psi_{LM_L}}_\H=1$, then \eqref{eq:thetaL-1} coincides
with
\begin{subequations}\label{eq:thetaL-2}
\begin{align}
E[\Psi_{LM_L}]\dfn&\int E_{LM_L}\norm{\Psi_{LM_L}(k)}_{\H(k)}^2dk,
\quad\text{where} \\
\Psi_{LM_L}(k)\dfn&\sqrt{\phi(k)}\psi_{LM_L}(k)\quad a.e.
\end{align}
\end{subequations}
The normalization of the vector field $\psi_{LM_L}$ implies \eqref{eq:infcnorm} and
\begin{equation}
0<\norm{\Psi_{LM_L}(k)}_{\H(k)}^2<1\quad a.e.
\label{eq:thetaL-3}
\end{equation}
By Theorem~\ref{thm:eigen}, for a.e. $k$
\begin{align*}
E_{LM_L}=&cL(L+1)-\Sigma[\Psi_{LM_L}(k)], \quad\text{where} \\
\Sigma[\Psi_{LM_L}(k)]\dfn&\sum_{J,\lambda}(-1)^{\lambda+1}\sqrt{2\lambda+1}
U_{\lambda L}(L0J\lambda k)\frac{c_{J\lambda^1\lambda LM_L}(k)}{c_{LM_L}(k)}
\end{align*}
since $c_{LM_L}(k)\neq0$ for a.e. $k$. Substitute the latter in \eqref{eq:thetaL-2}
and get that 
\begin{subequations}\label{eq:thetaL-4}
\begin{align}
E[\Psi_{LM_L}]=&cL(L+1)-\Sigma[\Psi_{LM_L}], \quad\text{where} \\
\Sigma[\Psi_{LM_L}]\dfn&\int\Sigma[\Psi_{LM_L}(k)]\norm{\Psi_{LM_L}(k)}_{\H(k)}^2dk.
\end{align}
\end{subequations}

The closure $\ol{\mathfrak{E}_L}$ consists of real numbers $r_{LM_L}$ such that, for 
every $\epsilon>0$, there exists a $E[\Psi_{LM_L}]\in\mathfrak{E}_L$ such that
$\abs{r_{LM_L}-E[\Psi_{LM_L}]}<\epsilon$. Put
\begin{subequations}\label{eq:thetaL-5}
\begin{align}
E_*[\Psi_{LM_L}]\dfn&cL(L+1)-\Sigma_*[\Psi_{LM_L}], \quad\text{where} \\
\Sigma_*[\Psi_{LM_L}]\dfn&\int\Sigma[\Psi_{LM_L}(k)]dk.
\end{align}
\end{subequations}
It follows from \eqref{eq:thetaL-3}, \eqref{eq:thetaL-4}, \eqref{eq:thetaL-5} that
\[cL(L+1),E_*[\Psi_{LM_L}]\in\ol{\mathfrak{E}_L}.\]

Assume that $\Sigma[\Psi_{LM_L}(k)]>0$ a.e. Then $E[\Psi_{LM_L}]>E_*[\Psi_{LM_L}]$
and $\inf\ol{\mathfrak{E}_L}$ is given by 
\begin{align*}
\inf E_*[\Psi_{LM_L}]=&cL(L+1)+\inf(-\Sigma_*[\Psi_{LM_L}]) \\
=&cL(L+1)-\Sigma_{LM_L}, \quad \Sigma_{LM_L}\dfn \sup\Sigma_*[\Psi_{LM_L}]
\end{align*}
where the infimum (resp. supremum) is taken over the points of 
$\Omega_{LM_L}(k)$ for a.e. $k$, and with the coordinates satisfying \eqref{eq:infcnorm}.

Assume that $\Sigma[\Psi_{LM_L}(k)]\leq0$ a.e. Then 
$cL(L+1)\leq E[\Psi_{LM_L}]\leq E_*[\Psi_{LM_L}]$ and hence $\inf\ol{\mathfrak{E}_L}=
cL(L+1)$.
Since $\inf\ol{\mathfrak{E}_L}$ is unique and $\inf\ol{\mathfrak{E}_L}=\E_{LM_L}$ by 
Lemma~\ref{lem:infnumL}, 
one deduces \eqref{eq:infeigen}. This completes the proof.
\end{proof}
\begin{rem}
It follows from the above proof that $\E<cL(L+1)$ is an exterior point of $\Theta_L$. In this
case $\E$ is not an eigenvalue of $A_L$. However, $\E$ is an eigenvalue of the discontinuous
part of $A_L$, \ie $\E$ belongs to the closure of the point spectrum of $A_L$.
\end{rem}
\begin{proof}[Proof of Corollary~\ref{cor:infnumL}]
According to \eqref{eq:infAL}, $A_{L}\geq\E$, and $\E$ is found from 
Theorem~\ref{thm:infnumL}.
Thus, using \eqref{eq:infeigen} and \eqref{eq:ceigen-b} and \eqref{eq:selfk}
\[\Sigma_{LM_L}=\sup\pv\Sigma_{L}^{(1)}(E,k)dk+\Sigma_{LM_L}^\prime.\]
We have the Cauchy principal value because, according to the definition \eqref{eq:selfk}
and \eqref{eq:znull}, we have to exclude the neighborhood of $k$ such that $\E\in \sigma_1(k)$ 
a.e. Now $E\geq\E$ a.e. implies
\begin{equation}
\sup\pv\Sigma_{L}^{(1)}(E,k)dk=\pv\Sigma_{L}^{(1)}(\E,k)dk
\label{eq:dhgf}
\end{equation}
and this shows \eqref{eq:infeigen2}.
\end{proof}
\begin{proof}[Proof of Corollary~\ref{cor:infnumL2}]
Eliminate 
$c_{\Lambda\lambda^2\Lambda00}(k)$ from the system \eqref{eq:cn3},
then substitute the obtained $c_{\lambda\lambda^1\lambda00}(k)$ in \eqref{eq:infeigen}
with $L=0$, and get that
\[\Sigma_{00}=\Sigma_{0}^{(1)}(\E)+\sup\pv\Sigma_{0}^{(2)}(E,k)dk
+\Sigma_{00}^{\prime\prime}\]
where we also use \eqref{eq:dhgf}, since $E\geq\E$ a.e. For the same reason,
the second integral on the right equals $\Sigma_{0}^{(2)}(\E)$, and we have
\eqref{eq:infeigen3}. Note that the Cauchy principal value arises because, according to the 
definition \eqref{eq:selfk2-2}, we have to exclude the neighborhoods of the $k$'s such that 
$\E\in \sigma_1(k)\mcup\sigma_2(k)\mcup\sigma_*(k)$ a.e.
\end{proof}
\section{Summary and concluding remarks}
In the present paper our goal was to propose a systematic and self-consistent way 
of dealing with higher-order phonon excitations induced by the impurity-boson interaction
potential. To achieve the goal, we found the direct integral description of angulon
most convenient, since in this representation we were able to work independently of
the initially unknown occupation numbers of phonon states. Subsequently, the
angular momentum algebra were less involved than it seemed to be from the beginning. Our 
results agree
with those obtained from the minimization of energy functional when the many-body
quantum state is based on an expansion in single-phonon excitations. We note that,
within the framework of single bath excitations, the variational approach coincides
with the Green function formalism, as defined by Feynman diagrams\footnote{Private
communication with M. Lemeshko.}. This indicates that the variational energy, and hence
our obtained infimum of the spectrum, is a "true" lowest energy.

Most challenging tasks arise when one deals with higher-order 
phonon excitations. Theorems~\ref{thm:eigen} and \ref{thm:infnumL} allow one to calculate 
the energies when finitely many phonon excitations are considered. The solutions of
equations in Theorem~\ref{thm:eigen} seem to be complicated in general, but in 
special cases they are elegant enough. As an example, we have found that
the two-phonon excitations cannot be ignored for a molecule in superfluid $^4$He. When the
impurity-bath interaction excites three and more phonons at the same time, one needs in
addition the values of the SCFPs as given in Tab.~\ref{tab:t1}; the other values are obtained
from the formulas derived in Sec.~\ref{sec:scfps-prop}.

We believe that the results reported in the paper can be further developed. Of course,
finding the solutions of Theorems~\ref{thm:eigen} would be enough to accomplish the
major part of the task, but there are less complicated, yet physically relevant, 
pieces of work to do as well. For example, one would like to say more about the spectral 
parts of the angulon operator, as modeled in the present paper. In this respect the resolvent
is of special interest. The analysis of the resolvent allows one to extend the present 
results of the eigenspace to the entire spectrum, including resonances.
\section*{Acknowledgments}
The author acknowledges productive discussions with Mikhail Lemeshko.


\end{document}